\documentclass[12pt]{iopart}
\usepackage{graphicx}
\usepackage{subcaption}
\usepackage{xcolor,cite}
\usepackage{amsthm}
\usepackage{amssymb}
\usepackage{bm}
\usepackage[utf8]{inputenc}

\begin{document}

\renewcommand{\theequation}{\thesection.\arabic{equation}}

\newtheorem{thm}{Theorem}[section]
\newtheorem{prp}[thm]{Proposition}
\newtheorem{lmm}[thm]{Lemma}
\newtheorem{cor}[thm]{Corollary}
\newtheorem{dfn}[thm]{Definition}
\newtheorem{rmk}[thm]{Remark}
\newtheorem{con}[thm]{Conjecture}
\newtheorem{Assum}[thm]{Assumption}
\newtheorem*{mth}{Main Theorem}

\newcommand{\vc}[1]{\mbox{\boldmath $#1$}}
\newcommand{\fracd}[2]{\frac{\displaystyle #1}
{\displaystyle #2}}

\def\ni{\noindent}
\def\nn{\nonumber}
\def\bH{\begin{Huge}}
\def\eH{\end{Huge}}
\def\bL{\begin{Large}}
\def\eL{\end{Large}}
\def\bl{\begin{large}}
\def\el{\end{large}}
\def\beq{\begin{eqnarray}}
\def\eeq{\end{eqnarray}}

\def\p{\partial}
\def\k{\kappa}

\def\C{{\cal C}}
\def\T{{\cal T}}
\def\H{{\cal H}}
\def\M{{\cal M}}
\def\L{{\cal L}}
\def\Q{{\cal Q}}
\def\S{{\cal S}}
\def\P{{\bf P}}
\def\im{{\rm Im}}
\def\re{{\rm Re}}

\def\i{{\rm i}}

\def\He{H\'enon }

\def\bn{\bigskip\noindent}
\def\pbn{\par\bigskip\noindent}
\def\n{\noindent}
\def\b{\bigskip}
\def\non{\nonumber}

\def\ve{\boldsymbol}

\newcommand{\bra}[1]{\langle{}#1{}|}
\newcommand{\eq}[1]{(\ref{#1})}


\title[Coupled Henon Map, Part I: Topological Horseshoes and Uniform Hyperbolicity]
{Coupled Henon Map, Part I: Topological Horseshoes and Uniform Hyperbolicity}

\author{Keisuke Fujioka$^{1}$, Ryota Kogawa$^{1}$, Jizhou Li$^{2}$ and Akira Shudo$^{1}$}
\address{$^{1}$ Department of Physics, Tokyo Metropolitan University, Tokyo 192-0397, Japan}

\address{$^{2}$ RIKEN iTHEMS, Wako, Saitama 351-0198, Japan}

\begin{abstract}
We derive a sufficient condition for topological horseshoe and uniform hyperbolicity 
of a 4-dimensional symplectic map, which is introduced by coupling 
the two 2-dimensional \He maps via linear terms. 
The coupled \He map thus constructed can be viewed as a simple map 
modeling the horseshoe in higher dimensions. 
We show that there are two different types of horseshoes, each of which 
is realized around different anti-integrable limits in the parameter regime. 

\end{abstract}
\maketitle

\color{black}

\section{Introduction}
\label{sec:introduction}

Horseshoe dynamics is known to be a source of chaos in dynamical systems. 
The most well-known and the simplest system modeling the horseshoe dynamics would be 
the H\'{e}non map~\cite{Henon69,Henon76}, which is a 2-dimensional quadratic map 
defined on ${\Bbb R}^2$. 
In the 2-dimensional plane, the horseshoe-shaped deformation is obtained 
by first stretching 
some initial domain in the unstable direction and then contracting it in the stable direction after folding back the stretched domain. 

Suppose that the horseshoe-shaped domain,  in both forward and backward iterations, 
intersects the original domain with two distinct regions, 
each of which is completely penetrated without lateral overhang. 
In this case, we say 
that the dynamics exhibits {\it topological horseshoe} \cite{Katok95,Wiggins88}. 
When the topological horseshoe is realized, 
the intersection of the iterated domain with the original domain, which generates 
the two disjoint strips in the case of a once-fold dynamics, is always mapped into the previous intersections, 
meaning that the width of each strip gradually decreases in time. 
Furthermore, if the contraction in the domain of interest is exponentially fast, each strip will eventually shrink to a string. 
If this is also the case in the backward iteration, the strings formed in the stretching and contracting directions iterates intersect to give a set of points. It then leads to the conjugation relation between the original and 
the properly introduced symbolic dynamics. 
The so-called Conley-Moser theory concerns a sufficient condition to have 
the symbolic dynamics based on topological horseshoe and {\it uniform hyperbolicity}~\cite{Wiggins03}. 

For the \He map, Devaney-Nitecki first developed such an argument and gave a sufficient 
condition such that the \He map exhibits topological horseshoe and uniform hyperbolicity as well~\cite{DN79}. 
Later, it was proved that the parameter locus satisfying uniform hyperbolicity
can be extended to the situation where the first homoclinic tangency happens using the complex 
dynamics technique~\cite{Bedford04} and computer-assisted proof~\cite{Arai07,Arai07_2}. 

There is another, even simpler approach to capturing the existence of chaos.
Suppose that the system has a certain parameter whose limiting value kills the dynamical relation between successive time steps, resulting in an infinite sequence of numbers or symbols. Such a limit is called the {\it anti-integrable limit}~\cite{Aubry90,Aubry94,Bolotin15}. 
Suppose there exists a suitable (discrete) Lagrangian. Then one can find a one-to-one correspondence between a sequence of numbers in the anti-integrable limit and an orbit generated by the actual dynamics whose parameter is close to the anti-integrable limit. 
The proof is based on the global implicit function theorem and 
the contraction mapping principle can be easily generalized to a wide class of systems. Moreover, since a close analogy exists between the orbits in dynamical systems and the equilibrium states of a class of variational problems in solid-state systems, one can relate the uniform hyperbolicity of the dynamics with the existence of phonon gap 
in the solid state problem~\cite{Aubry92}. 

The topic we would like to discuss in this article is the topological horseshoe and uniform hyperbolicity in 
higher dimensional symplectic maps. 
Among a variety of choices~\cite{Mao88,Howard87,Ding90,Bountis94,Todesco94,Vrahatis96,Todesco96,Gemmi97,Vrahatis97,giovannozzi98,Richter14,Lange14,Onken16,Lange16,Anastassiou17,Backer18,Backer20}, we here take the coupled \He map, which will be introduced below.
As in the case of 2-dimensional polynomial maps~\cite{Friedland89}, 
there is a derivation of normal forms of quadratic symplectic maps 
due to Moser~\cite{Moser94}, which provides a canonical model to
be studied in detail~\cite{Backer18,Backer20}. 
Indeed, it was shown in \cite{Backer20} that the normal form introduced by Moser 
can be decoupled into a pair of uncoupled quadratic maps under an appropriate choice of 
parameters, so our map should be a reduced version of the general normal form. 

An advantage of starting with the coupled \He map to examine 
topological horseshoes and uniform hyperbolicity would be that 
one can find anti-integrable limits in the parameter space rather easily.
As mentioned above, one would expect uniform hyperbolicity, and perhaps also 
topological horseshoe as well in the vicinity of anti-integrable limits~\cite{Qin01,Juang08,Chen15,Chen16,Hampton22}. 
There indeed exist some works in which topological horseshoe together with uniform hyperbolicity 
manifests in the region close to the anti-integrable limit~\cite{Du06,Li06}.

Here we provide a sufficient condition for topological horseshoe and uniform hyperbolicity for the coupled \He map, 
using essentially the same strategy as Devaney-Nitecki~\cite{DN79}. 
%
In particular, we study topological horseshoe and uniform hyperbolicity around the two different anti-integrable limit, 
each of which is derived by taking certain parameter limits in the coupled \He map. The first type can be shown to be conjugate with the symbolic dynamics with four symbols, while the second one is described by the full shift with two symbols. As will be briefly explained and thoroughly discussed in the following paper, their folding natures are different from each other. Especially the first type is so unique that it appears only in 4-dimensional space.

The structure of the paper is as follows:
Section \ref{sec:CH} introduces our coupled \He map, 
which is obtained by coupling a pair of 2-dimensional \He maps, 
and has three parameters: the two nonlinearity parameters and the coupling strength. 
Then we show that two anti-integrable limits exist in the current form of the coupled \He map. 
Section \ref{sec:Main_theorem} gives the main results of this paper, providing a sufficient condition 
for topological horseshoe and uniform hyperbolicity around each anti-integrable limit. 
Section \ref{sec:non-wandering_set} presents the existence domains 
in which the non-wandering set is contained. 
This part corresponds to the proof of the first part of the main theorems. 
Section \ref{sec:sufficient_condition_for_hyperbolicity} 
gives a sufficient condition for uniform hyperbolicity of the coupled \He map. 
To derive uniform hyperbolicity, we examine the cone field condition. 
In particular, we will use a sufficient condition for uniform hyperbolicity in higher dimensional 
settings, which have been introduced by Newhouse~\cite{Newhouse04}. 
Section \ref{sec:proof_Main_theorem} is devoted to proving the main theorems. 
Section \ref{sec:summary} summarizes the results and provides some outlooks.

\setcounter{equation}{0}
\section{Coupled H\'enon map and anti-integrable limits}
\label{sec:CH}

\subsection{Coupled H\'enon map}

The coupled H\'enon map is introduced as 
\begin{eqnarray}
\label{eq:coupled_Henon_map}
  \left(\begin{array}{c}
    x_{n+1}\\
    y_{n+1}\\
    z_{n+1}\\
    w_{n+1}
  \end{array}\right)
  =f  \left(\begin{array}{c}
    x_n\\
    y_n\\
    z_{n}\\
    w_{n}
    \end{array}\right)
  =  \left(\begin{array}{c}
    a_0-x_n^2-z_n+c(x_n-y_n)\\
    a_1-y_n^2-w_n-c(x_n-y_n)\\
    x_{n}\\
    y_{n}
  \end{array}\right) ,
\end{eqnarray}
where $c>0$ is assumed. 
The inverse map $f^{-1}$ is 
\begin{eqnarray}
  \left(\begin{array}{c}
    x_{n-1}\\
    y_{n-1}\\
    z_{n-1}\\
    w_{n-1}
  \end{array}\right)
  =f^{-1}\left(\begin{array}{c}
    x_n\\
    y_n\\
    z_n\\
    w_n
  \end{array}\right)
  =\left(\begin{array}{c}
    z_n\\
    w_n\\
    a_0-z_n^2-x_n+c(z_n-w_n)\\
    a_1-w_n^2-y_n-c(z_n-w_n)\\
  \end{array}\right) .
\end{eqnarray}
Here $a_0$ and $a_1$ are parameters that control the nonlinearity, and 
the parameter $c$ gives the coupling strength between the two H\'enon maps~\cite{DN79}. 
For $c> 0$, the replacement of the variables as $(x,y,z,w)\rightarrow (z,w,x,y)$ 
transforms the map $f$ into its inverse $f^{-1}$. 
\begin{eqnarray}
	\left(\begin{array}{c}
	X\\
	Y\\
	Z\\
	W
	\end{array}\right)
	=\frac{1}{2}\left(\begin{array}{c}
	x+y\\
	x-y\\
	z+w\\
	z-w\\
	\end{array}\right), 
	\label{eq:change_of_coordinate}
\end{eqnarray}
the form (\ref{eq:coupled_Henon_map}) can be written as
\begin{eqnarray}
\label{eq:map_tranformed}
	\left(\begin{array}{c}
    X_{n+1}\\
    Y_{n+1}\\
    Z_{n+1}\\
    W_{n+1}
  \end{array}\right)
  =F\left(\begin{array}{c}
    X_n\\
    Y_n\\
    Z_n\\
    W_n
  \end{array}\right)
  =\left(\begin{array}{c}
    A_0-(X_n^2+Y_n^2)-Z_n\\
    A_1-2X_nY_n-W_n+2cY_n\\
    X_n\\
    Y_n
  \end{array}\right). 
\end{eqnarray}
where $$\displaystyle A_0=\frac{a_0+a_1}{2}, ~~~~~A_1=\frac{a_0-a_1}{2}. $$
The inverse map $F^{-1}$ is also rewritten as 
\begin{eqnarray}
\label{eq:map_tranformed_inverse}
  \left(\begin{array}{c}
    X_{n-1}\\
    Y_{n-1}\\
    Z_{n-1}\\
    W_{n-1}
  \end{array}\right)
  =F^{-1}\left(\begin{array}{c}
    X_n\\
    Y_n\\
    Z_n\\
    W_n
  \end{array}\right)
  =\left(\begin{array}{c}
    Z_n\\
    W_n\\
    A_0-(Z_n^2+W_n^2)-X_n\\
   A_1-2Z_nW_n-Y_n+2cW_n\\
  \end{array}\right) .
\end{eqnarray}
%
%

\subsection{Anti-integrable limits for the coupled H\'enon map}

Here we show that two types anti-integrable limits exist in the coupled H\'enon map. 
For simplicity, we consider the case with $a=a_0=a_1$. 
Let us introduce new parameters $\epsilon=\sqrt{1/a}, u=\epsilon x$ and $v=\epsilon y$ and 
rewrite the coupled H\'enon map (\ref{eq:coupled_Henon_map}) as 
\begin{eqnarray}
	\left\{\begin{array}{c}
		\epsilon u_{n+1}=1-(u_{n})^2-\epsilon u_{n-1}+c\epsilon(u_n-v_n),\\
		\epsilon v_{n+1}=1-(v_{n})^2-\epsilon v_{n-1}-c\epsilon(u_n-v_n).
	\end{array}\right.\label{eq:AI}
\end{eqnarray}

\bn
(A) Anti-integrable limit with four symbols

The first type of anti-integrable limit is given by 
letting $a\rightarrow \infty$ with $c$ being fixed. 
In this anti-integrable limit, the coupling between two H\'enon maps can be neglected 
and (\ref{eq:AI}) tends to 
\begin{eqnarray}
	\left\{\begin{array}{c}
		0=1-(u_n)^2,\\
		0=1-(v_n)^2,
	\end{array}\right.
\end{eqnarray}
which lead to 
\begin{eqnarray}
	\left\{\begin{array}{c}
		u_n=\pm 1,\\
		v_n=\pm 1.
	\end{array}\right.
\end{eqnarray}
The four solutions $(u_n,v_n)=(+1,+1),(+1,-1),(-1,+1),(-1,-1)$ provide symbols of the 
symbolic dynamics around this anit-integrable limit.

\bn
(B) Anti-integrable limit with two symbols

The second type of anti-integrable limit is given by letting $a\rightarrow \infty$
with $c/\sqrt{a}=const=\gamma$ being fixed. In this limit, the two H\'enon maps 
are strongly coupled and the relations (\ref{eq:AI}) tend to 
\begin{eqnarray}
	\left\{\begin{array}{c}
		0=1-(u_n)^2+\gamma(u_n-v_n),\\
		0=1-(v_n)^2-\gamma(u_n-v_n),
	\end{array}\right.
\end{eqnarray}
which lead to the four solutions in the form
$(u_n,v_n)=(+1,+1),(\gamma-\sqrt{1-\gamma^2},\gamma+\sqrt{1-\gamma^2}),(\gamma+\sqrt{1-\gamma^2},\gamma-\sqrt{1-\gamma^2})$ and $(-1,-1)$. 
For $1\le |\gamma|$, the two solutions are complex, 
while for $1 > |\gamma|$ all the solutions are real. 

\section{Main theorems}
\label{sec:Main_theorem}

In this paper, we will give 
 sufficient conditions for topological horseshoe and 
uniform hyperbolicity around the anti-integrable limits (A) and (B), respectively. 

\begin{thm} 
\label{thm:MainA}
As for the anti-integrable limit of the case (A), the following holds. \\
	A-1) For $-1\le A_0$, the non-wandering set $\Omega(f)$ satisfies
	\begin{equation}
		\Omega(f)\subset  V_f, 
	\end{equation}
	where
	\begin{eqnarray}
		&V_f=\{(x,y,z,w)\,|\, |x|,|y|,|z|,|w|\le r\}. 
	\end{eqnarray}
	Here, $r=2\sqrt{2}(1+\sqrt{1+A_0})$. 
	\\
	A-2) If the parameters satisfy the following conditions, $f$ shows topological horseshoe. 
		  \begin{eqnarray}
		  \label{eq:typeA_sufficient1}
  			&0<\frac{1}{4}c^2+a_i-(c+2)r,\ \ (i=0,1), \\
		  \label{eq:typeA_sufficient2}
			&0\le r^2-2(c+1)r-a_i,\ \ (i=0,1). 
		\end{eqnarray}
	A-3) In addition to the conditions (\ref{eq:typeA_sufficient1}) and  (\ref{eq:typeA_sufficient2}), if the parameters satisfy the following condition, $\Omega(f)$ is uniformly hyperbolic. 
		  \begin{eqnarray}
			4+c<\frac{-c+\sqrt{c^2+4(a_i-(c+2)r)}}{2}, \ \ (i=0,1).
		\end{eqnarray}
	\label{the:CHA}
\end{thm}

\color{black}

\begin{thm}
\label{thm:MainB}
As for the anti-integrable limit of the case (B), the following holds.\\
	B-1) For $-1\le A_0$, the non-wandering set $\Omega(f)$ satisfies
	\begin{equation}
		\Omega(f)\subset V_F, 
	\end{equation}
	where
	\begin{eqnarray}
	\label{eq:def_V_F}
		&V_F=\{(x,y,z,w)\,|\, \left|\frac{x+y}{2}\right|,\left|\frac{x-y}{2}\right|,\left|\frac{z+w}{2}\right|,\left|\frac{z-w}{2}\right|\le R\}, 
	\end{eqnarray}
	where $\displaystyle R=1+\sqrt{1+A_0}$.  
	\\
	B-2) If the parameters satisfy the following conditions, $f$ shows topological horseshoe. 
		\begin{eqnarray}
		  \label{eq:typeB_sufficient1}
			&A_1 \le R<c,\\
		  \label{eq:typeB_sufficient2}
			&R<A_0-(W^*)^2-R,\\
		  \label{eq:typeB_sufficient3}
			&W^*\le R. 
		\end{eqnarray}
		Here, $W^*=\max\Bigl(\displaystyle \Bigl|\frac{2R-A_1}{2(c-R)}\Bigr|$, $\Bigl|\displaystyle\frac{-2R-A_1}{2(c-R)}\Bigr|\Bigr)$, and $Z^*=\sqrt{A_0-(W^*)^2-2R}$.  \\
	B-3) In addition to the conditions (\ref{eq:typeB_sufficient1}), (\ref{eq:typeB_sufficient2}) and  (\ref{eq:typeB_sufficient3}), if the parameters satisfy the following condition, $\Omega(f)$ is uniformly hyperbolic. 
		\begin{eqnarray}
\label{the:CHBe}
			4+c\le Z^*-W^*.
		\end{eqnarray}
	\label{the:CHB}
\end{thm}

\color{black}

\setcounter{equation}{0}
\section{Non-wandering set}
\label{sec:non-wandering_set}

\subsection{Some lemmas}
\label{sec:lemmas}
To prove topological horseshoe and uniformly hyperbolicity 
for the coupled H\'enon map
we take a similar strategy similar to Devaney-Nitecki~\cite{DN79}. 
In the following, we prove some lemmas using 
the parameter: 
\begin{eqnarray}
  R&=&1+\sqrt{1+A_0} \in {\Bbb R}. 
\end{eqnarray}
\begin{lmm}
$R$ satisfies the following
  \begin{eqnarray}
    R^2-2R-A_0=0. 
  \end{eqnarray}
  \label{lem:BasicEq1}
\end{lmm}
\begin{proof}
  Self-evident. 
\end{proof}
\begin{lmm}
a) For any $C\ge 0$, if $|Z_0|\le C$ is satisfied, the following holds:
\begin{eqnarray}\label{lem:BasicInEQ}
  A_0-(Z_1^2+W_1^2)-C\le X_1 \le A_0-(Z_1^2+W_1^2)+C. 
  \end{eqnarray}
In addition,  if $|X_0|\le C$, then $|Z_1|\le C$ holds. \\
b) For any $C\ge 0$, if $|X_0|\le C$ is satisfied, the following holds:
\begin{eqnarray}
  A_0-(X_{-1}^2+Y_{-1}^2)-C\le Z_{-1} \le A_0-(X_{-1}^2+Y_{-1}^2)+C. 
\end{eqnarray}
In addition, if $|Z_0|\le C$, then $|X_{-1}|\le C$ holds.
    \label{lem:horseshoe}
\end{lmm}
\begin{proof}
It is easy to check both of them. 
\end{proof}
\begin{lmm}
  a) If $X_0\le \min(-|Z_0|,-R)$, then $X_{1}\le X_{0}$ follows. The equality holds when 
  $(X_0,Y_0,Z_0)=(-R,0,-R)$. \\
  b) If $-|Z_{0}| \le X_0$ and $Z_{0}\le -R$ hold, then $Z_{-1}\le Z_{0}$ 
  and
  $|Z_{0}|\le |Z_{-1}|$ follows.
  The equalities hold when $(X_0,Z_0,W_0)=(-R,-R,0)$. 
  \label{lem:Diff1}
\end{lmm}
\begin{proof}
\begin{figure}
        \centering
        \includegraphics[width=0.3\linewidth]{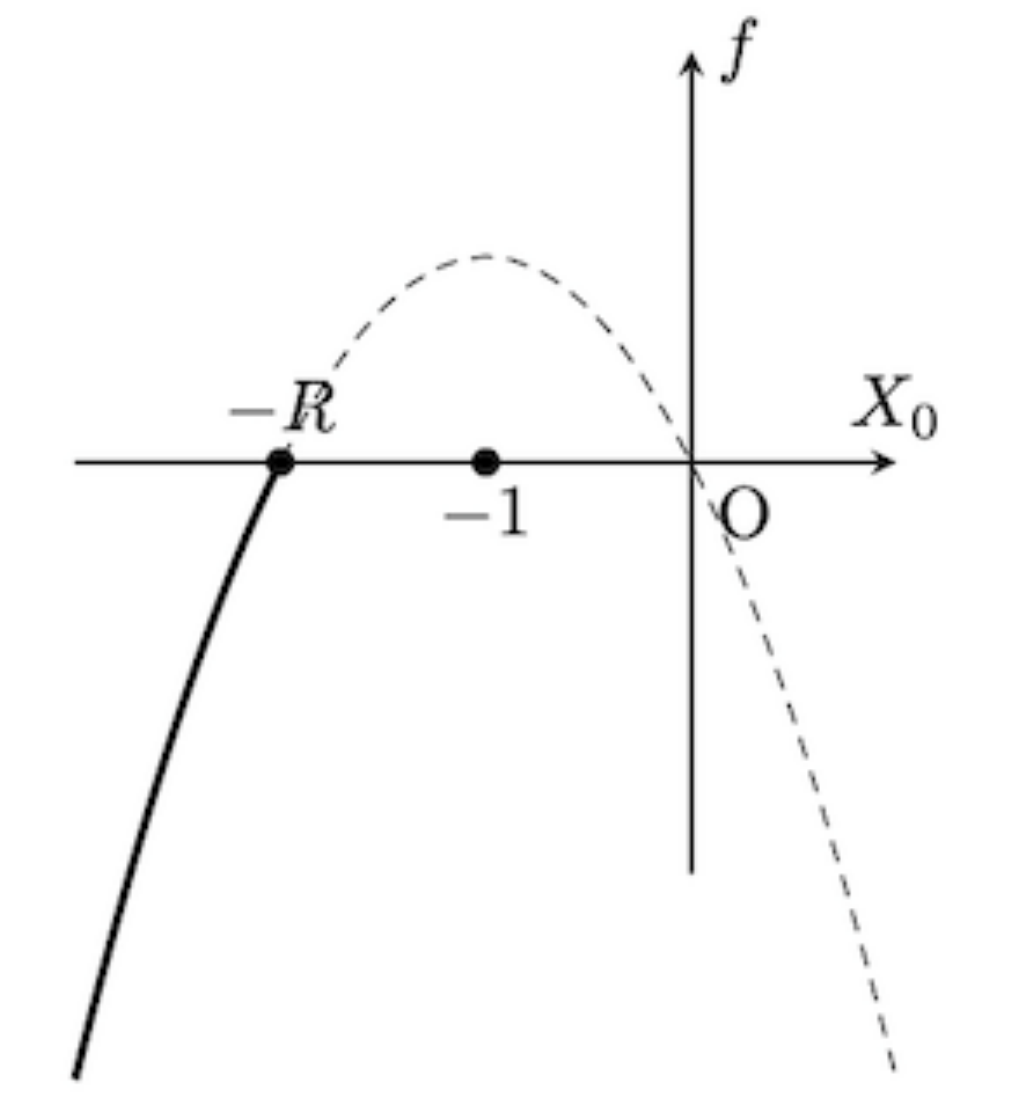} 
        \caption{Sketch of $f=A_0-(X_0)^2-2X_0$.} \label{fig:DiffEq1.1}
\end{figure}
a) 
If $X_{0}\le \min(-|Z_{0}|,-R)$ holds, we find that 
  \begin{eqnarray}
    X_1-X_0&=& A_0-(X_0^2+Y_0^2)-Z_0-X_0\nonumber\\
    &\le& A_0-X_0^2-Z_0-X_0\nonumber\\
    &\le& A_0-X_0^2+|Z_0|-X_0\nonumber\\
    &\le& A_0-X_0^2-2X_0\nonumber\\
    &=&A_0-(X_0+1)^2+1. \label{eq:DiffIneq1.1}
  \end{eqnarray}
  Since $X_0\le -R$, we have a condition for $X_0$ as
  \begin{eqnarray}
    X_0\le -R= -1-\sqrt{1+A_0}\le -1. 
  \end{eqnarray}
  Then, $A_0-(X_0+1)^2+1$ takes the maximum value at $X_0=-R$
  (see Fig.~\ref{fig:DiffEq1.1}).  Thus, we have 
  \begin{eqnarray}
    X_1-X_0\le A_0-(-R)^2-2(-R)=0.
  \end{eqnarray}
  Here we have used lemma \ref{lem:BasicEq1}. The equality holds when $(X_0,Y_0,Z_0)=(-R,0,-R)$ is satisfied. 

b) Assuming $-|Z_{0}| \le X_0$ and $Z_0 \le -R$, 
  we find that 
  \begin{eqnarray}
    Z_{-1}-Z_0&=&A_0-(Z_0^2+W_0^2)-X_0-Z_0\nonumber\\
    &\le&A_0-Z_0^2-X_0-Z_0\nonumber\\
    &\le&A_0-Z_0^2+|Z_0|-Z_0\nonumber\\
    &=&A_0-Z_0^2-2Z_0\nonumber\\
    &=&A_0-(Z_0+1)^2+1. 
    \label{eq:DiffEq1.2}
  \end{eqnarray}
  In the same way as above,  since 
  \begin{eqnarray}
    Z_0\le -R\le -1 
  \end{eqnarray}
 holds,  $A_0-(Z_0+1)^2+1$ takes the maximum value at $Z_0=-R$ 
 (see Fig.~\ref{fig:DiffEq1.2}). 
Hence, we have 
  \begin{eqnarray}
   Z_{-1}-Z_0\le A_0-(-R)^2-2(-R)=0. \label{eq:DiffEq1.2con}
  \end{eqnarray}
Since $Z_0 \leq -R$, $|Z_{-1}|\ge |Z_0|$ also follows.
The equality holds when $(X_0,Z_0,W_0)=(-R,-R,0)$ holds. 
  \begin{figure}
    \centering
	\includegraphics[width=0.3\linewidth]{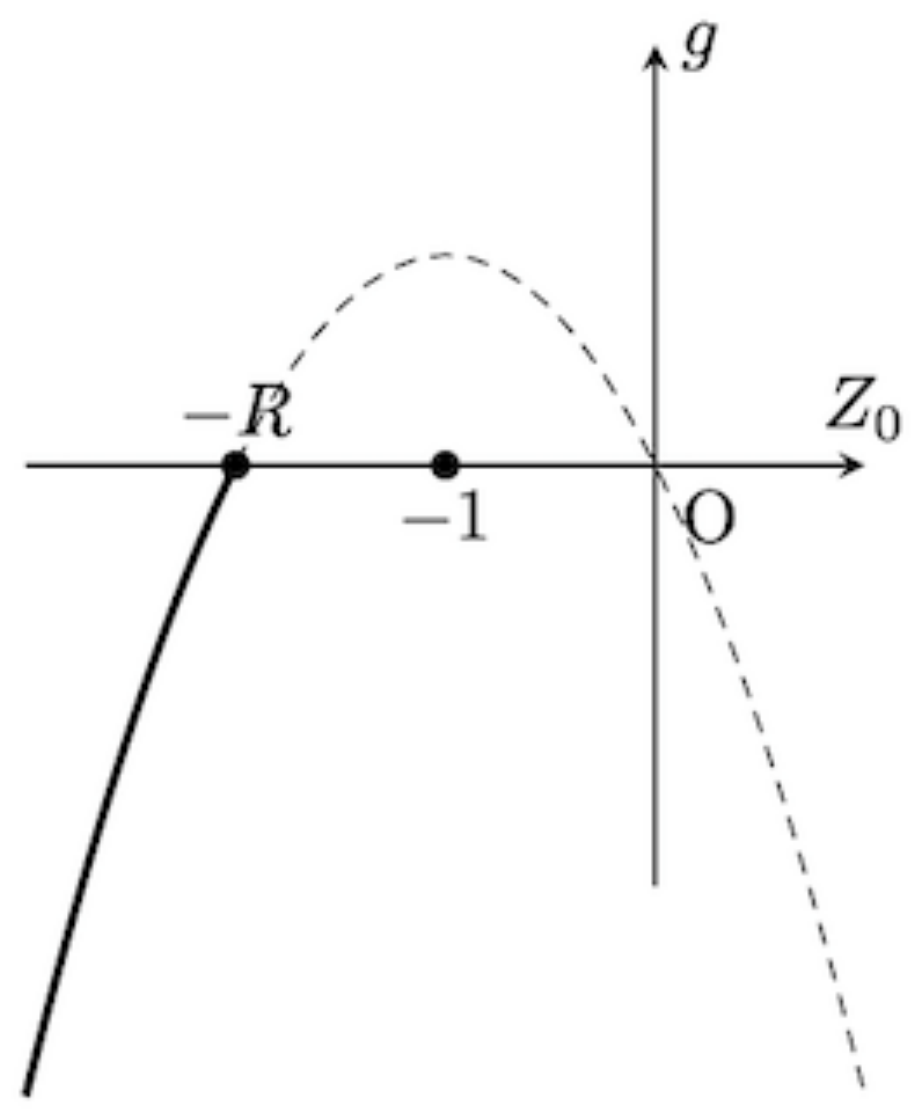}
    \caption{Sketch of $g=A_0-(Z_0)^2-2Z_0$.}
    \label{fig:DiffEq1.2}
  \end{figure}
\end{proof}


\subsection{Decomposition of domains and transition rules}

In the following, we study the coupled H\'enon map in the case 
where $R$ takes a real value. 
For this purpose, we introduce the following domains (see Fig.~\ref{fig:Domain1}):
\begin{eqnarray}
  N_1&=&\{(X,Y,Z,W)\, |\, X\le \min(-|Z|,-R)\}, \\
  N_2&=&\{(X,Y,Z,W)\,|\,X\ge -R,\ |Z|\le R\}, \\
  N_3&=&\{(X,Y,Z,W)\,|\,X\ge -|Z|,\ Z\ge R\}, \\
  N_4&=&\{(X,Y,Z,W)\,|\,X\ge -|Z|,\ Z\le -R\}, \\
\end{eqnarray}
\begin{figure}
  \centering
	\includegraphics[height=6cm, bb = 0 0 345 322]{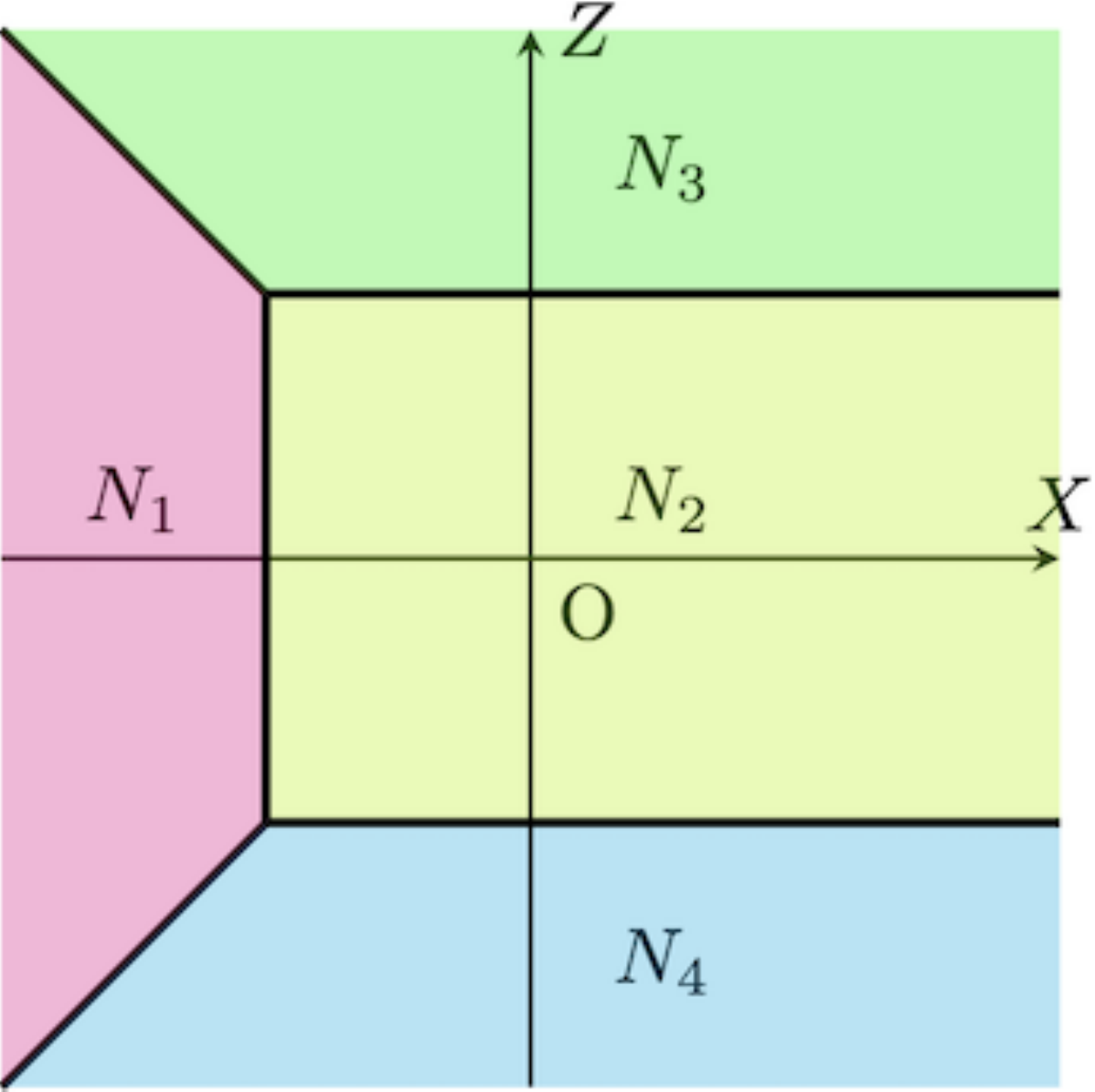}
  \caption{Illustration of domains and their boundary lines
  in the $(X,Z)$-plane. }
  \label{fig:Domain1}
\end{figure}

\begin{prp}
  If  $A_0 \ge -1$, the following holds: \\
    a) Under the iteration of $F$, the coordinate $X$ strictly decreases in $N_1$ except for $(X,Z)=(-R,-R)$. \\
  b) $F(N_1)\subset N_1$.\\
  c) $F(N_2)\subset N_1\cup N_2$ and $F(N_3)\subset N_1\cup N_2$.\\
  d) Under the iteration of $F^{-1}$, the coordinate $Z$ strictly decreases in $N_4$ except for $(X,Z)=(-R,-R)$. \\
  e) $F^{-1}(N_3)\subset N_4$ and $F^{-1}(N_4)\subset N_4$. \\
\color{black}
  f) $F^{-1}(N_2) \subset N_2\cup N_3\cup N_4$.\\
\color{black}
  \label{pro:NonWand1}
\end{prp}
\begin{proof}
  a) Self-evident from lemma \ref{lem:Diff1}\,a). \\
  b) For $(X_0,Y_0,Z_0,W_0)\in N_1$, $X_1\le X_0$ follows from a). 
  From Eq.~(\ref{eq:map_tranformed}),  we have 
  $Z_1=X_0\le-R < 0$. So, we have $-|Z_1|=X_0$, which leads to the inequality $X_1\le -|Z_1|$. In addition, $X_1\le X_0\le -R$ follows from a). 
  Combining these, $(X_1,Y_1,Z_1,W_1)\in N_1$ is satisfied, 
  {\it i.e.}, $F(N_1)\subset N_1$ holds. \\
  c) Using lemma \ref{lem:horseshoe} by setting $C=R$, the region specified 
  by $|Z| \le R$, which covers
  the domain $N_2$,  is mapped to the horseshoe-shaped domain (see Fig.~\ref{fig:Horseshoe1.0}): 
  \begin{eqnarray}
  \label{eq:horseshoe-shaped}
  A_0-(Z_1^2+W_1^2)-R\le X_1 \le A_0-(Z_1^2+W_1^2)+R. 
  \end{eqnarray}
 The right boundary is expressed as
  \begin{eqnarray}
    X_1=A_0-(Z_1^2+W_1^2)+R  . \label{eq:UpF1} 
  \end{eqnarray}
  Since $W_1$ is real, $X_1$ is bounded as 
   \begin{eqnarray}
   \label{eq:UpF1_ineq}
    X_1\le A_0-Z_1^2+R. 
  \end{eqnarray}
The boundary of Eq.~(\ref{eq:UpF1_ineq}), namely, 
\begin{eqnarray}
\label{eq:UpF1_eq}
  X_1= A_0-Z_1^2+R, 
\end{eqnarray}
is shown by the red curve in Fig.~\ref{fig:Horseshoe1.0}. 
Using lemma \ref{lem:BasicEq1}, it is easy to check that the point $(X_1,Z_1)=(-R,\pm R)$ satisfies Eq.~(\ref{eq:UpF1_eq}). 
Therefore, the horseshoe-shaped region, specified 
by Eq.~(\ref{eq:horseshoe-shaped}), 
lies completely inside the left-hand side of the red curve expressed by Eq.~(\ref{eq:UpF1_eq}), and the red curve passes through the conner points 
$(X_1, Z_1) = (-R, \pm R)$ of $N_2$. Thus, 
$F(N_2)\subset N_1\cup N_2$ is concluded. 
  \begin{figure}
    \centering
    	\includegraphics[height=4cm,bb = 0 0 675 322]{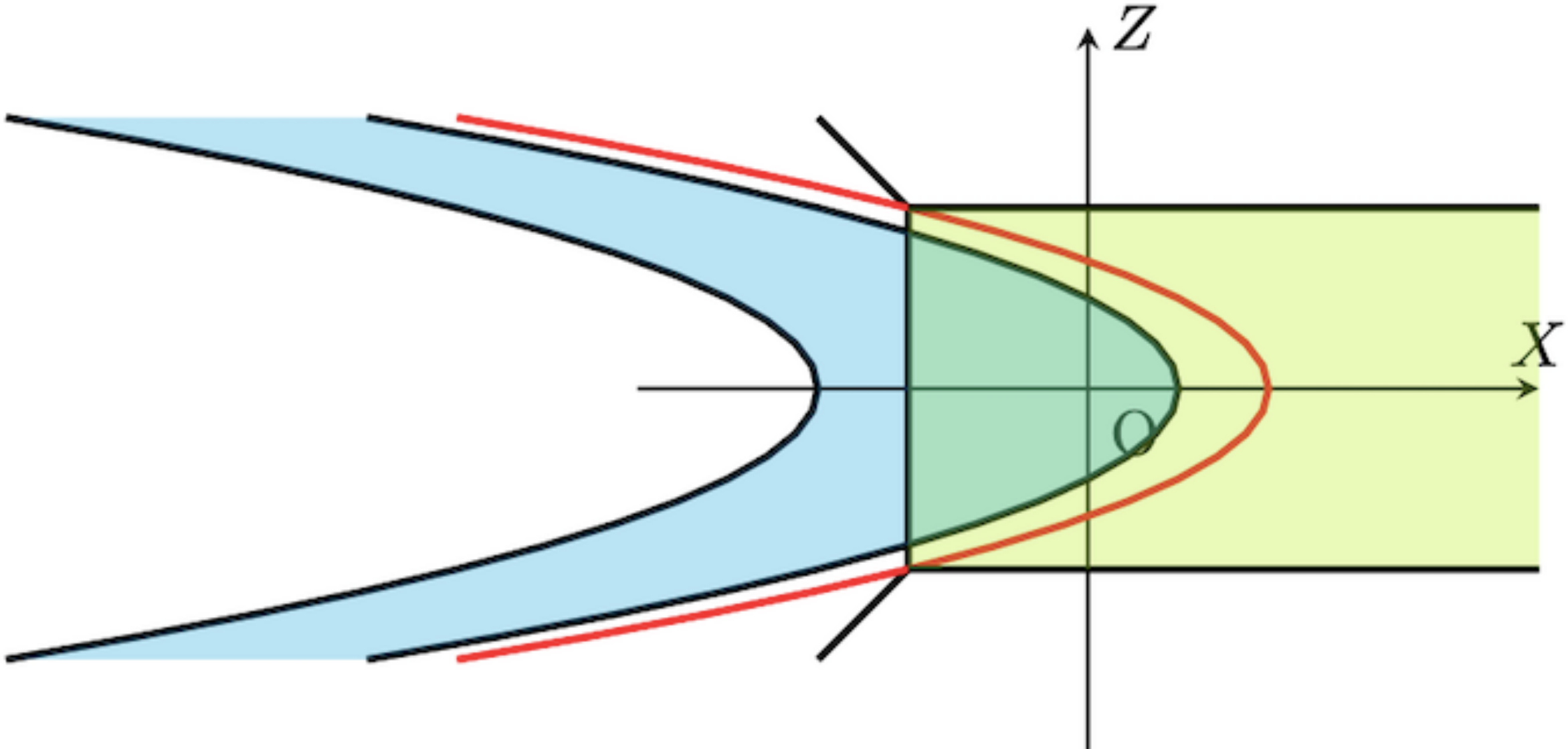}
    \caption{The blue region shows the horseshoe region obtained by setting $C=R$ in 
(\ref{lem:BasicInEQ}). 
    The red curve is the parabola in Eq.~(\ref{eq:UpF1_eq}). 
    }
    \label{fig:Horseshoe1.0}
  \end{figure}
   It is also easy to show that the line $Z=R$ is mapped to the leftmost curve 
  shown in Fig.~\ref{fig:Horseshoe1.0}. 
  As a result, $F(N_3)\subset N_1\cup N_2$ follows
  (see Fig.~\ref{fig:FDomain1}). 
  \begin{figure}
    \centering
		\includegraphics[width=0.7\linewidth]{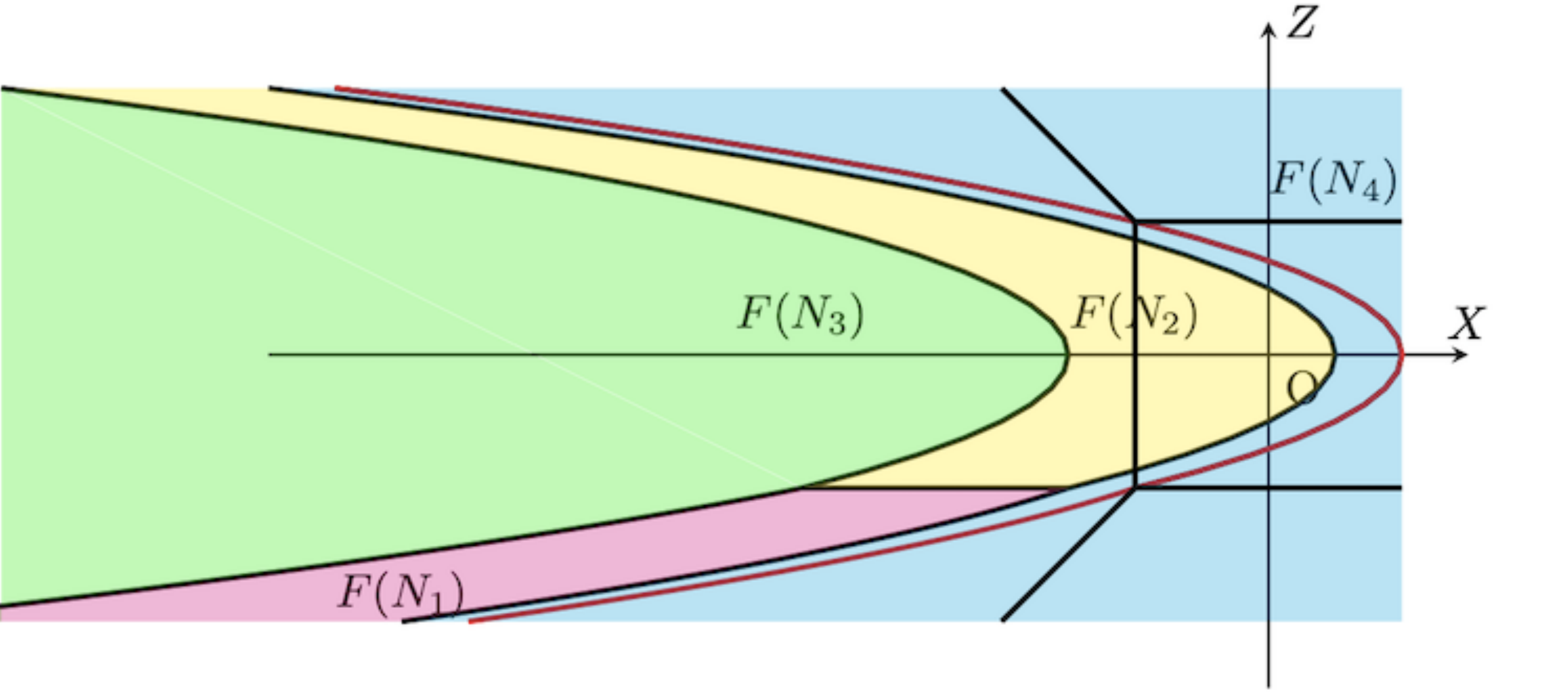}
    \caption{The iterated domains. The red curve represents the rightmost curves for 
    the regions $F(N_1),F(N_2)$ and $F(N_3)$.}
    \label{fig:FDomain1}
  \end{figure}
\noindent
d) Self-evident from lemma \ref{lem:Diff1}\,b). \\
e) Note that the domain $N_3$ in $(X,Z)$-plane is expressed as 
\begin{eqnarray}
N_3= \{(X,Y,Z,W)\, | \, Z=-X + \gamma, \, \gamma \ge 0, \, Z \ge R \}, 
\end{eqnarray} 
thus it is mapped by $F^{-1}$ as
\begin{eqnarray}
\fl
~~~~
F^{-1}(N_3) = \{(X,Y,Z,W)\, | \, Z=A_0-(X^2+Y^2)+X-\gamma, \, 
\gamma \ge 0, \, X \ge R \}. 
\end{eqnarray} 
For $\gamma \ge 0$, we find that 
$Z=A_0-(X^2+Y^2)+X-\gamma \le A_0-X^2+X \le -R$. 
Here we have used lemma \ref{lem:BasicEq1} and 
$X \ge R$ in the region $F^{-1}(N_3)$. 
Since the points in $F^{-1}(N_3)$ satisfy $X \ge R$ and $Z \le -R$, $F^{-1}(N_3) \subset N_4$ holds. 

In a similar way, the domain $N_4$ is expressed as 
\begin{eqnarray}
N_4= \{(X,Y,Z,W)\, | \, Z= X-\gamma, \,  Z \le -R, \, \gamma \ge 0\}, 
\end{eqnarray} 
thus it is mapped by $F^{-1}$ as
\begin{eqnarray}
\fl
~~~~
F^{-1}(N_4) = \{(X,Y,Z,W)\, | \, Z=A_0-(X^2+Y^2)-X-\gamma, 
\, 
\gamma \ge 0, \, X \le -R \}. 
\end{eqnarray} 
For $\gamma \ge 0$, we find that 
$Z-X=A_0-(X^2+Y^2)-X-\gamma-X \le A_0-X^2- 2X \le 0$. 
Here we have  again used lemma \ref{lem:BasicEq1} and 
$X \le -R$ in the region $F^{-1}(N_4)$. 
Since the points in $F^{-1}(N_4)$ satisfy 
$X \le -R$ and $Z \le X$, $F^{-1}(N_4) \subset N_4$ holds. 

\color{black}

\noindent
  f) It is easy to see that
  $F^{-1}(N_2)$ is contained 
  in the region $|X|\le R$. Combining these facts with the definitions of 
  $N_1,N_2$ and $N_3$, one can show that 
  $F^{-1}(N_2)\in N_2\cup N_3 \cup N_4$ holds (see Fig.~\ref{fig:InFDomain1}). 
  \end{proof}
  \begin{figure}
	\centering
    	\includegraphics[width=0.4\linewidth]{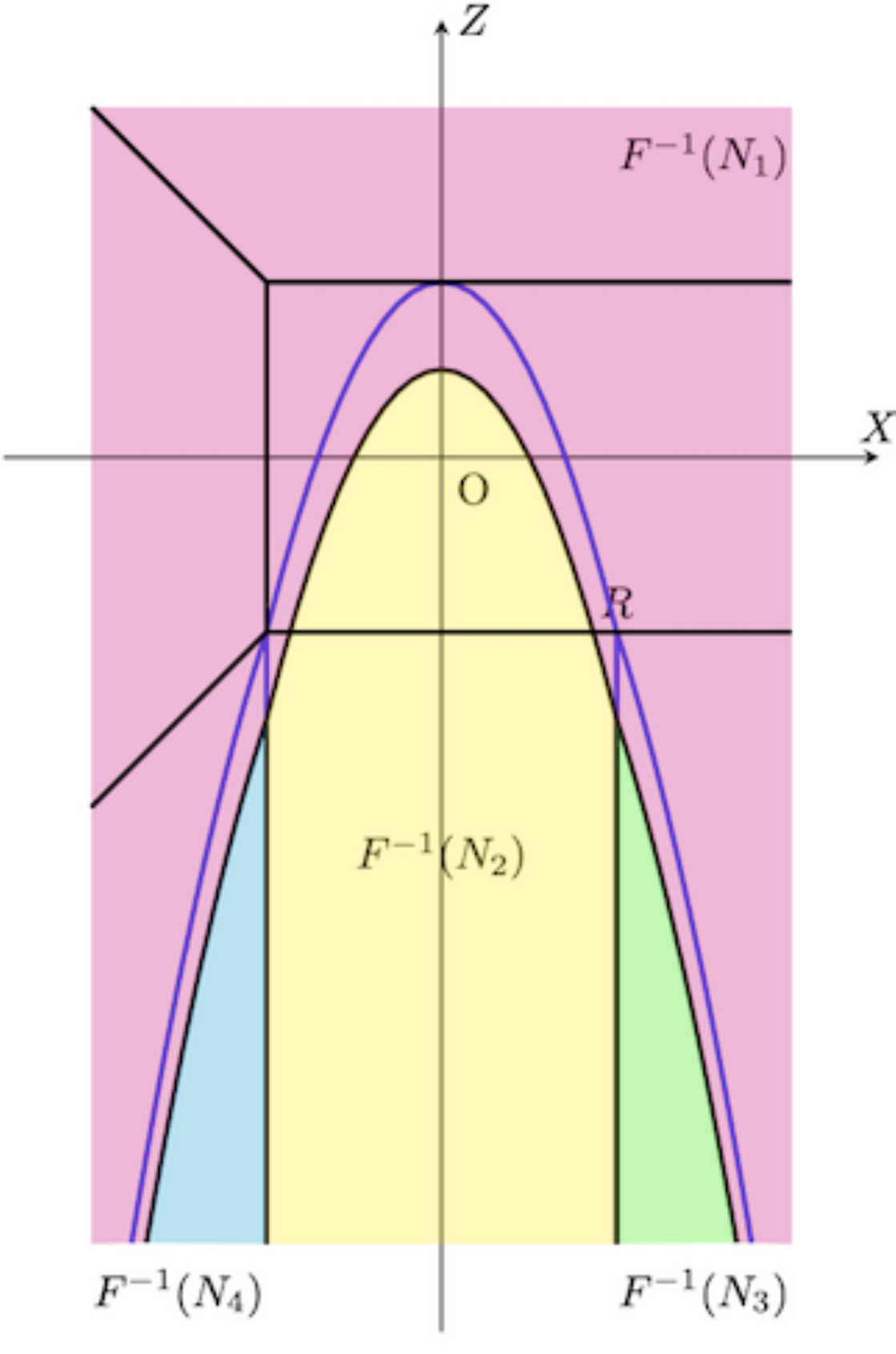}
     	\caption{The inverse images of each domain. The same set of parameters is used as in Fig. \ref{fig:FDomain1}. The blue curve represents the uppermost situation. }
     \label{fig:InFDomain1}
  \end{figure}


\subsection{Existence domain of the non-wandering set: proof of Main theorem \ref{thm:MainA} A-1) and Main theorem \ref{thm:MainB} B-1)}
\color{black}

In this section, based on Propositon \ref{pro:NonWand1}, 
we specify the domain containing the non-wandering set $\Omega(F)$. 
As illustrated in Fig.~\ref{fig:NonWand1} the flow of dynamics, 
regardless of whether the flow from $N_4$ to $N_2$ exists or not, 
\color{black}
an orbit launched in the domain $N_3$ does not return back to the vicinity of the initial point. 
Propositons \ref{pro:NonWand1} {a) implies that the coordinate $X$ of the points contained in $N_1$
are strictly decreasing, and also \ref{pro:NonWand1} d) implies the coordinate $Z$ of the points contained in $N_4$ are strictly increasing, 
\color{black}
so they do not return back to the vicinity of the initial points as well. 
This argument holds for the backward iteration $F^{-1}$. 
It follows that the points of the non-wandering set $\Omega(F)$ do not exist 
in the domains $N_1,N_3$ and $N_4$, and thus 
the non-wandering set $\Omega(F)$ should be contained in the domain $N_2$.

  \begin{figure}
        \centering
	\includegraphics[width=0.4\linewidth]{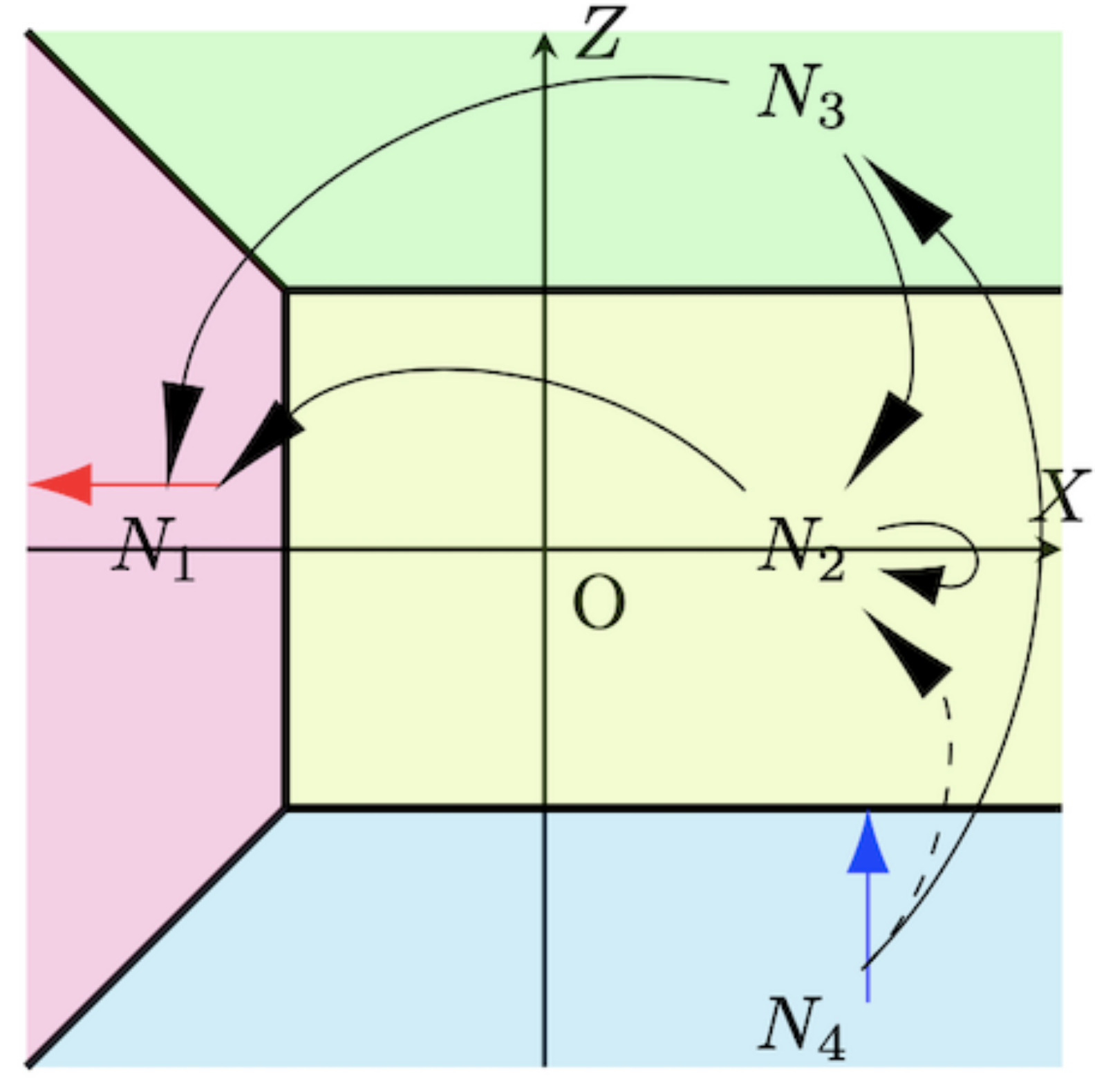}
\caption{The flow of dynamics. 
The dashed line shows that the flow can exist, but its proof 
is not given here. 
The red and blue arrows indicate 
a monotonic shift to the left and upward, respectively, in each region.
}
\label{fig:NonWand1}
   \end{figure}
   
   Since $\Omega(F)\subset N_2$, the non-wandering set can be 
   expressed as $\Omega(F)=\Lambda \subset \bigcap_{k=-\infty}^{\infty}{F^k(N_2)}$, 
   thus $\Omega(F)\subset F^{-1}(N_2)\cap N_2 \cap F(N_2)$ holds.
   Note here that we do not know whether the non-wandering set is empty or not.
   In the following, we use this condition to further specify the existence domain 
   of the non-wandering set.  More specifically, we will provide a hypercube 
   containing the region $F^{-1}(N_2)\cap N_2 \cap F(N_2)$.

First, note that the non-wandering set $\Omega(F)$ 
should be
located in the region $|Z|\le R$, 
since $\Omega(F)\subset F^{-1}(N_2)\cap N_2 \cap F(N_2)$.
From the mapping rule (\ref{eq:map_tranformed_inverse}), 
$|Z_0|\le R$ immediately leads to $|X_{-1}|\le R$. 
Therefore, the condition $\Omega(F)\subset F^{-1}(N_2)$ implies 
that $|X|\le R$ must be satisfied for the points in $\Omega(F)$
(see Fig.~\ref{fig:InFDomain1}). 

Next, we recall (\ref{eq:UpF1}), which tells us the maximum value of $X$ 
in the region $F(N_2)$, that is, 
\begin{eqnarray}
     X_1&=&A_0-(Z_1^2+W_1^2)+R\nonumber\\
     &\le& A_0-W_1^2+R. 
\end{eqnarray}
The condition $|X_1| \le R$, obtained above, leads to 
\begin{eqnarray}
-R \le A_0 - W_1^2 + R, 
\end{eqnarray}
which implies that $|W_1| \le R$ must be satisfied for the points in $\Omega(F)$
 (see Fig.~\ref{fig:DownFInF1}(a)). 
 Again, it follows immediately from the mapping rule 
(\ref{eq:map_tranformed_inverse}) 
that $|Y_0|\le R$. 

As a result of these arguments, we can conclude that 
\begin{eqnarray}
\Omega(F)\subset V_F=\{ (X,Y,Z,W) \, |\, |X|,|Y|,|Z|,|W|\le R\}. 
\end{eqnarray}
We then consider the hypercube $V_f$ in the original coordinates $(x,y,z,w)$, 
which contains the region $V_F$. 
The slice of $V_f$ by $(x,y)$-plane is illustrated in Fig.~\ref{fig:Circle1}, and we have
\begin{eqnarray}
\Omega(f)\subset V_f=\{ (x,y,z,w)\, |\,|x|,|y|,|z|,|w|\le 2\sqrt{2}R\}. 
\end{eqnarray}
The proof of our Main theorems \ref{thm:MainA} A-1) and \ref{thm:MainB} B-1) is thus 
completed. 
\color{black}

\color{black}

\begin{figure}
    \centering
	    \begin{tabular}{c}
	    \hspace{-68pt}
      \begin{minipage}{0.5\hsize}
        \centering
	\includegraphics[width=1.2\linewidth]{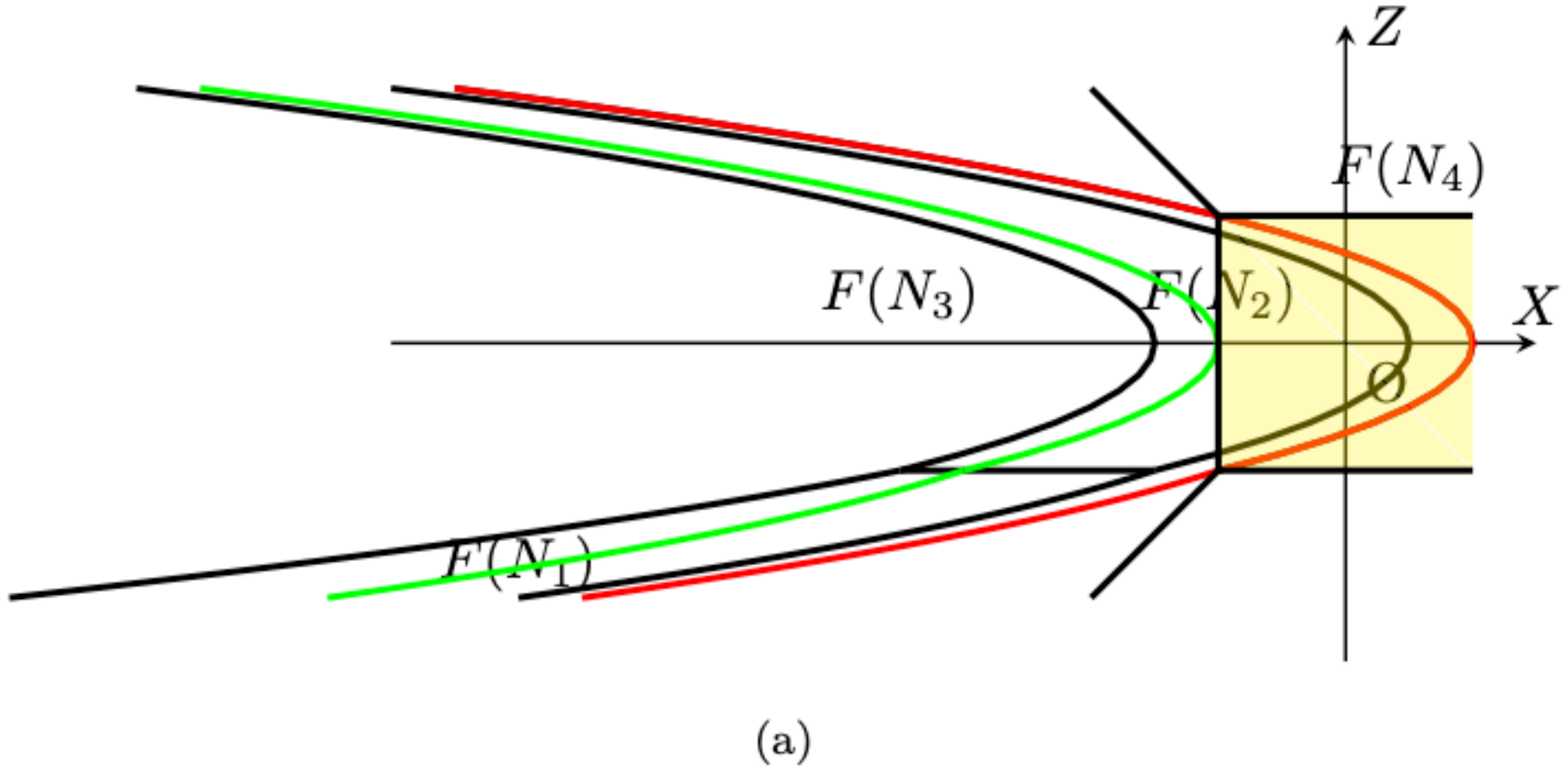}
      \end{minipage} 
      \begin{minipage}{0.5\hsize}
        \centering
	\includegraphics[width=0.5\linewidth]{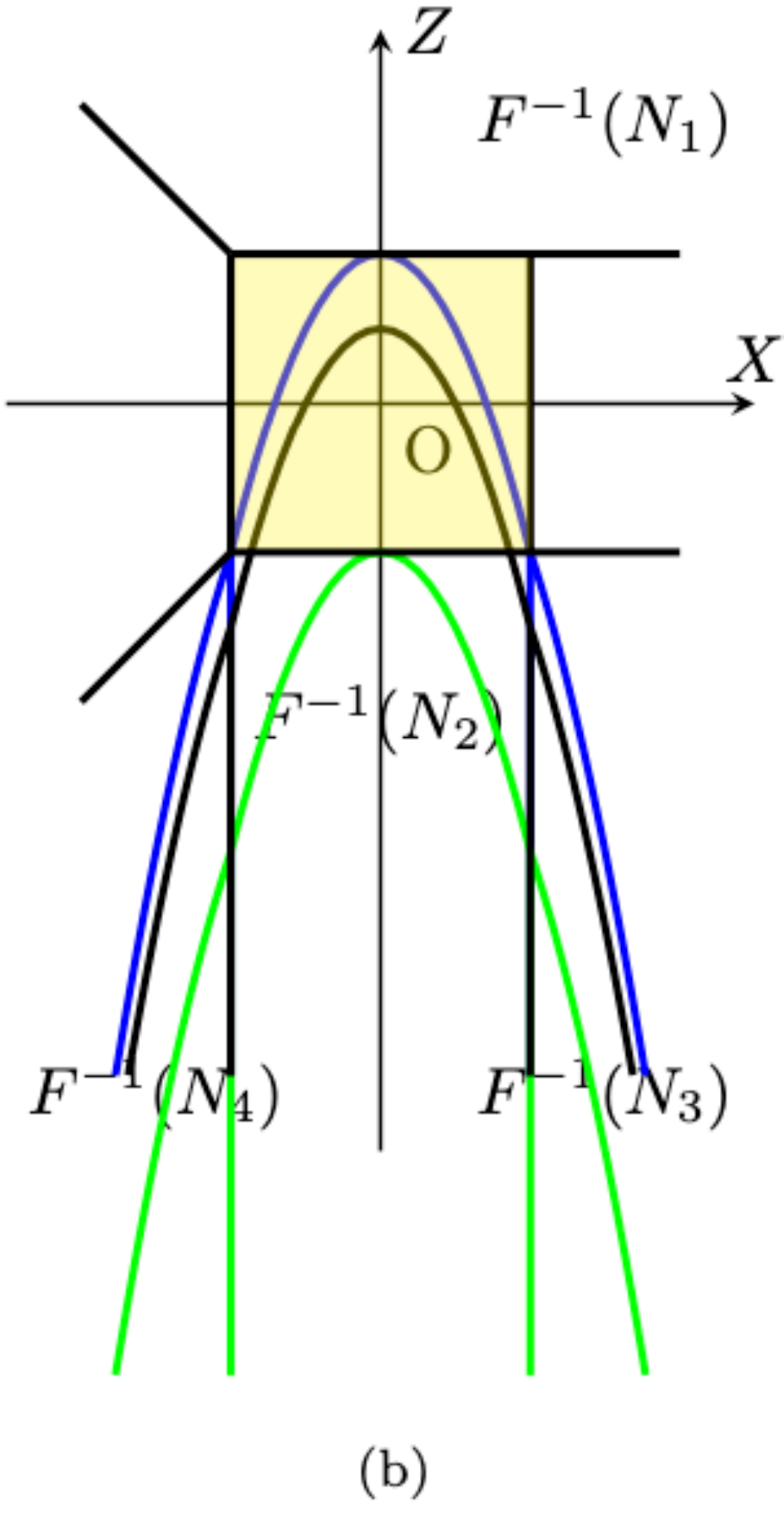}
	 	\vspace{+0pt}
      \end{minipage}\\
    \end{tabular}
    \caption{The domains mapped by $F$ and $F^{-1}$. 
    (a) The green curve shows the leftmost parabola for which
    $N_2 \cap F(N_2) \neq \emptyset$. (b) The green curve shows the
     lowest parabola for which
    $N_2 \cap F^{-1}(N_2) \neq \emptyset$. 
    }
    \label{fig:DownFInF1}
  \end{figure}
\begin{figure}
    \centering
	\includegraphics[width=6cm,bb = 0 0 329 322]{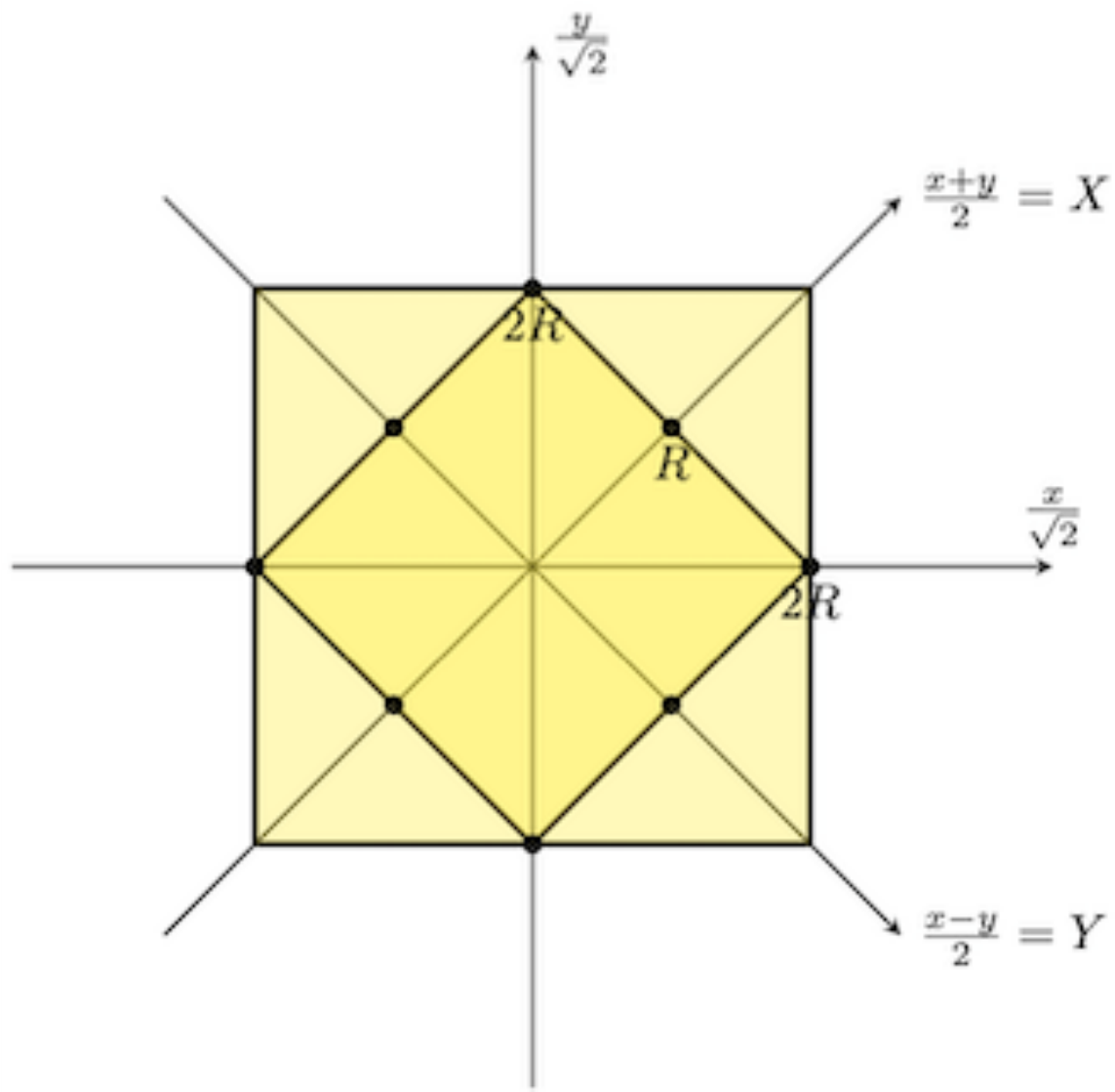}
    \caption{Domains containing the non-wandering set $\Omega(F)$.}
    \label{fig:Circle1}
  \end{figure}

\setcounter{equation}{0}
\section{Sufficient condition for uniform hyperbolicity}
\label{sec:sufficient_condition_for_hyperbolicity}

\subsection{Cone field condition}
We introduce here the cone field condition~\cite{Newhouse04}, which leads to a sufficient 
condition for uniform hyperbolicity. 
\begin{dfn}
Let $\mathbb{E}_1\subset \mathbb{R}^n$ and $\mathbb{E}_2$ be a proper subspace and 
its complementary subspace, respectively.  i.e., 
$\mathbb{R}^n = \mathbb{E}_1\oplus \mathbb{E}_2$. 
The standard unit cone determined by the subspaces $\mathbb{E}_1$ and $\mathbb{E}_2$
is given by the set, 
	\begin{eqnarray}
		K(\mathbb{E}_1,\mathbb{E}_2)=\{\bm{v}=(\bm{v}_1,\bm{v}_2)\,|\bm{v}_1\in\mathbb{E}_1,\bm{v}_2\in\mathbb{E}_2,|\bm{v}_2|\le |\bm{v}_1|\}. 
	\end{eqnarray}
\end{dfn}
\begin{dfn}
A cone in $\mathbb{R}^n$ with core $\mathbb{E}_1$, denoted by $\mathcal{C}(\mathbb{E}_1)$, is the image $T(K(\mathbb{E}_1,\mathbb{E}_2))$. Here $T : \mathbb{R}^n \to \mathbb{R}^n$ is a
linear automorphism such that $T(\mathbb{E}_1) = \mathbb{E}_1$. 
By a cone $\mathcal{C}$ in $\mathbb{R}^n$ 
we mean a set $\mathcal{C}(\mathbb{E}_1)$ 
for some proper subspace $\mathbb{E}_1$ of $\mathbb{R}^n$. 
\end{dfn}
\begin{dfn}
A cone field $\mathcal{C}=\{\mathcal{C}_{\bm{x}}\}$ on a manifold $M$
is a collection of cones $\mathcal{C}_{\bm{x}} \in T_xM$ for $x \subset M$.
\end{dfn}
\begin{dfn}
For a given cone field $\mathcal{C}=\{\mathcal{C}_{\bm{x}}\}_{\bm{x}\in M}$ and 
a diffeomorphism $h$ defined on the manifold $M$, 
let 
	\begin{eqnarray}
		m_{\mathcal{C},\bm{x}}&=&m_{C,\bm{x}}(h)=\inf_{\bm{v}\in{\mathcal{C}}_{\bm{x}}\backslash\{0\}}\frac{|Dh_{\bm{x}}(\bm{v})|}{|\bm{v}|}, \\
		m'_{\mathcal{C},\bm{x}}&=&m'_{C,\bm{x}}(h)=\inf_{\bm{v}\notin{\mathcal{C}}_{h(\bm{x})}}\frac{|Dh^{-1}_{h({\bm{x}})}(\bm{v})|}{|\bm{v}|}. 
	\end{eqnarray}
We call $m_{\mathcal{C},\bm{x}}$ and $m'_{\mathcal{C},\bm{x}}$
the minimal expansion and minimal co-expansion 
of $h$ on $\mathcal{C}_{\bm{x}}$, respectively. 
\begin{dfn}
We say that $h$ is expanding on the cone field $\mathcal{C}$ if 
\begin{eqnarray}
\inf_{\bm{x} \in \Lambda} m_{\mathcal{C},\bm{x}}(h) > 1 
~
\Longleftrightarrow
~
\inf_{x \in \Lambda} \inf_{\bm{v} \in \mathcal{C}_{\bm{x}}\backslash\{0\}}
\frac{|Dh_{\bm{x}}(\bm{v})|}{|\bm{v}|} > 1. 
\end{eqnarray}
Similarly, we say that $h$ is co-expanding on the cone field $\mathcal{C}$ if 
\begin{eqnarray}
\inf_{\bm{x} \in \Lambda} m'_{\mathcal{C},\bm{x}}(h) > 1 
~
\Longleftrightarrow
~
\sup_{x \in \Lambda} \sup_{\bm{u} \in Dh^{-1}_{h({\bm{x})}}(C_{h({\bm{x}})}^c) }
\frac{  |Dh_{\bm{x}}(u)|}{|\bm{u}|} < 1. 
\end{eqnarray}
\end{dfn}
\end{dfn}
\begin{dfn}
We say that the cone
field $\mathcal{C}_{\bm{x}}$ has constant orbit core dimension on $\Lambda$ if 
	\begin{eqnarray}
		\dim{\mathbb{E}_{\bm{x}}}=\dim{\mathbb{E}_{h(\bm{x})}}
	\end{eqnarray}
holds for all $x\in \Lambda$. Here $\mathbb{E}_{\bm{x}}$ and 
$\mathbb{E}_{h(\bm{x})}$ are the cores of 
$\mathcal{C}_{\bm{x}}$ and $\mathcal{C}_{h(\bm{x})}$,  respectively. 
\end{dfn}
Based on these notions, Newhouse has derived a necessary and sufficient condition 
for uniform hyperbolicity. 
\begin{thm}[Newhouse]
A sufficient condition for $\Lambda(h)$ to be uniformly hyperbolic is 
that there are an integer $N>0$ and a cone field $\mathcal{C}$ with constant orbit core dimension over $\Lambda(h)$ such that $h^N$ is both expanding and co-expanding on $\mathcal{C}$.
\end{thm}
Here we can show the following. 
\begin{cor}
If there exists a standard unit cone field $\mathcal{C}_{\bm{x}}$ on $\Lambda(h)$ with 
$h$-invariant cones, i.e., $Dh(\mathbb{E}_{\bm{x}}) = \mathbb{E}_{h(\bm{x})}, \, \forall x \in \Lambda(h)$, 
such that $h$ is both expanding and co-expanding, then 
$\Lambda(h)$ is uniformly hyperbolic. 
\label{the:SuffConefield}
\end{cor}
\begin{proof}
Since $\mathbb{E}_{\bm{x}}$ is invariant under $h$, it has constant orbit core dimension. 
The fact that for any $\bm{x} \in \Lambda(h)$ $\lambda\le m_{\mathcal{C},\bm{x}}$ and $\lambda\le m'_{\mathcal{C},\bm{x}}$ imply that $h$ is both expanding and co-expanding. Hence $h$ is uniformly hyperbolic. 
\\
\end{proof}

\subsection{Sufficient condition for uniform hyperbolicity: the case with four symbols in the anti-integrable limit}
\label{sec:uniform_hyp_4}
\label{sec:CHHyper4}
We first derive a sufficient condition for the case whose anti-integrable limit has 
four symbols. 
The Jabcobian for the forward and backward iterations is 
respectively given by 
\begin{eqnarray}
  &Jf&=\left(\begin{array}{cccc}
    -2x+c&-c&-1&0\\
    -c&-2y+c&0&-1\\
    1&0&0&0\\
    0&1&0&0
  \end{array}\right), \\
  &Jf^{-1}&=\left(\begin{array}{cccc}
    0&0&1&0\\
    0&0&0&1\\
    -1&0&-2z+c&-c\\
    0&-1&-c&-2w+c
  \end{array}\right). 
\end{eqnarray}

\n
The following lemma will be used in the subsequent argument. 
\begin{lmm}
Let 
	\begin{eqnarray}
		  G(x,y)=\left(\begin{array}{cc}
			-2x+c&-c\\
			-c&-2y+c
  		\end{array}\right), 
	\end{eqnarray}
where $x,y \in {\Bbb R}$ satisfy the condition $2\lambda+2+c\le |x|,|y|$. 
Then, for any vector $\bm{w}_0= (\xi,\eta)^t$, the following holds:
	\begin{eqnarray}
		(2\lambda+2)|\bm{w}_0|\le |\bm{w}_1|, 
	\end{eqnarray}
	where $\bm{w}_1=G(x,y) \bm{w}_0$. 
	\label{lem:CHHyper4}
\end{lmm}
\begin{proof}
In the case $|\eta_0|\le|\xi_0|$, we have
\begin{eqnarray}
|\bm{w}_1|&\ge& |\xi_1|\nonumber\\
	&=&|(-2x+c)\xi_0-c\eta_0|\nonumber\\
	&\ge&|-(2x-c)\xi_0|-|c\eta_0|\nonumber\\
	&=&|(2x-c)||\xi_0|-c|\eta_0|\nonumber\\
	&\ge&(2|x|-c)|\xi_0|-c|\eta_0|\nonumber\\
	&\ge&2(|x|-c)|\xi_0|\nonumber\\
	&=&(|x|-c)(|\xi_0|+|\xi_0|)\nonumber\\
	&\ge&(|x|-c)(|\xi_0|+|\eta_0|)\nonumber\\
	&\ge&(|x|-c)|\bm{w}_0|\nonumber\\
	&\ge&(2\lambda+2)|\bm{w}_0|. \nonumber
\end{eqnarray}
Similarly, for $|\xi_0|<|\eta_0|$, 
\begin{eqnarray}
|\bm{w}_1|&\ge& |\eta_1|\nonumber\\
	&=&|-c\xi_0+(-2y+c)\eta_0|\nonumber\\
	&\ge&|-(2y-c)\eta_0|-|c\xi_0|\nonumber\\
	&=&|(2y-c)||\eta_0|-c|\xi_0|\nonumber\\
	&\ge&(2|y|-c)|\eta_0|-c|\xi_0|\nonumber\\
	&\ge&2(|y|-c)|\eta_0|\nonumber\\
	&=&(|y|-c)(|\eta_0|+|\eta_0|)\nonumber\\
	&\ge&(|y|-c)(|\xi_0|+|\eta_0|)\nonumber\\
	&\ge&(|y|-c)|\bm{w}_0|\nonumber\\
	&\ge&(2\lambda+2)|\bm{w}_0|. \nonumber
\end{eqnarray}
\end{proof}
Suppose that $\mathbb{E}^+,\mathbb{E}^-\subset \mathbb{R}^2$ 
gives ${\mathbb{R}}^4=\mathbb{E}^-\oplus\mathbb{E}^+$. 
For $K=K(\mathbb{E}^+,\mathbb{E}^-)$, we can show the following. 
\begin{thm}
Suppose that the matrix $G(x,y)$ satisfies the condition, 
	\begin{eqnarray}
		(2\lambda+2)|\bm{v}|\le |G(x,y)\bm{v}|, 
	\end{eqnarray}
for some $\lambda>1$ and any vector $\bm{v}(\in {\mathbb{R}^2})$. 
In addition, $Jf$ and $Jf^{-1}$ are expressed in terms of 
the $2\times 2$ identify matrix $I_2$ and the zero matrix $O_2$ as 
	\begin{eqnarray}
	  	&Jf&=\left(\begin{array}{cc}
    			G(x,y)&-I_2\\
			I_2&O_2
  		\end{array}\right), \\
  		&Jf^{-1}&=\left(\begin{array}{cc}
			O_2&I_2\\
			-I_2&G(z,w)
  		\end{array}\right), 
 	\end{eqnarray}
then the following holds:\\
a) For any vector $\bm{v}_0\in K$, $\lambda|\bm{v}_0|\le |\bm{v}_1|$ holds where $\bm{v}_1=Jf(\bm{v}_0)$.  \\
b) For any vector $\bm{v}_0\notin K$, $\lambda|\bm{v}_0|\le |\bm{v}_{-1}|$ holds where 
$\bm{v}_{-1}=Jf^{-1}(\bm{v}_0)$. 
	\label{the:CHHyper}
\end{thm}
\begin{proof}
For a), we have 
	\begin{eqnarray}
		|\bm{v}_1|&=&\left|\left(\begin{array}{c}
    			G(x,y)\bm{v}_0^+-\bm{v}_0^-\\
    			\bm{v}_0^+
  		\end{array}\right)\right|\nonumber\\
		&\ge&|G(x,y)\bm{v}_0^+-\bm{v}_0^-|-|\bm{v}_0^+|\nonumber\\
		&\ge&|G(x,y)\bm{v}_0^+|-|\bm{v}_0^-|-|\bm{v}_0^+|\nonumber\\
		&\ge&(2\lambda+1)|\bm{v}_0^+|-|\bm{v}_0^-|\nonumber\\
		&\ge&2\lambda|\bm{v}^+_0|\nonumber\\
		&\ge&\lambda(|\bm{v}^+_0|+|\bm{v}^-_0|)\nonumber\\
		&\ge&\lambda|\bm{v}_0|. \nonumber
	\end{eqnarray}
	Similarly, for b), we have 
	\begin{eqnarray}
		|\bm{v}_{-1}|&=&\left|\left(\begin{array}{c}
    			\bm{v}_0^-\\
    			-\bm{v}_0^++G(z,w)\bm{v}_0^-
  		\end{array}\right)\right|\nonumber\\
		&\ge&|G(z,w)\bm{v}_0^--\bm{v}_0^+|-|\bm{v}_0^-|\nonumber\\
		&\ge&|G(z,w)\bm{v}_0^-|-|\bm{v}_0^+|-|\bm{v}_0^-|\nonumber\\
		&\ge&(2\lambda+1)|\bm{v}_0^-|-|\bm{v}_0^+|\nonumber\\
		&\ge&2\lambda|\bm{v}^-_0|\nonumber\\
		&\ge&\lambda(|\bm{v}^+_0|+|\bm{v}^-_0|)\nonumber\\
		&\ge&\lambda|\bm{v}_0|. \nonumber
	\end{eqnarray}
\end{proof}
Theorem \ref{the:CHHyper} tells us that $f$ is expanding and co-expanding. 
Combined with the Lemma \ref{lem:CHHyper4}, we finally find the following:
\begin{cor}
\label{cor:uniform_4}
If all points in the non-wandering set $\Omega(f)$, if not empty, 
satisfy the condition 
\begin{eqnarray}
	4+c\le |x|,|y|,|z|,|w|, \label{eq:CHHyper4}
\end{eqnarray}
then $\Omega(f)$ is uniformly hyperbolic. 
\end{cor}
\color{black}

\subsection{Sufficient condition for uniformly hyperbolicity: the case with two symbols in the anti-integrable limit}
\label{sec:uniform_hyp_2}
\label{sec:CHHyper2}
Next, we consider a sufficient condition for the case where the anti-integrable limit has 
two symbols. 
The Jacobian after the transformation (\ref{eq:change_of_coordinate}) 
is respectively given by 
	\begin{eqnarray}
	  	&Jf&=\left(\begin{array}{cc}
    			\widetilde{G}(x,y)&-I_2\\
			I_2&O_2
  		\end{array}\right), \\
  		&Jf^{-1}&=\left(\begin{array}{cc}
			O_2&I_2\\
			-I_2&\widetilde{G}(z,w)
  		\end{array}\right). 
 	\end{eqnarray}
The following will be used in the following argument. 
\begin{lmm}
Let 
	\begin{eqnarray}
		  \widetilde{G}(X,Y)=\left(\begin{array}{cc}
			-2X&-2Y\\
			-2Y&-2X+2c
  		\end{array}\right), 
		\label{eq:CHHyper2G}
	\end{eqnarray}
where $X,Y\in {\Bbb R}$ satisfy the condition $2\lambda+2+c\le |X|-|Y|$. 
Then, for any vector $\bm{w}_0= (\xi,\eta)^t$, the following holds:
	\begin{eqnarray}
		(2\lambda+2)|\bm{w}_0|\le |\bm{w}_1|, 
	\end{eqnarray}
where $\bm{w}_1=G(x,y)\bm{w}_0$. 
	\label{lem:CHHyper2}
	\end{lmm}
\begin{proof}
In the case $|\eta_0|\le|\xi_0|$, we have
\begin{eqnarray}
|\bm{w}_1|&\ge& |\xi_1|\nonumber\\
	&=&|(-2X)\xi_0-2Y\eta_0|\nonumber\\
	&\ge&|-2X\xi_0|-|2Y\eta_0|\nonumber\\
	&\ge&2(|X|-|Y|)|\xi_0|\nonumber\\
	&=&(|X|-|Y|)(|\xi_0|+|\xi_0|)\nonumber\\
	&\ge&(||X|-|Y|)(|\xi_0|+|\eta_0|)\nonumber\\
	&\ge&(|X|-|Y|)|\bm{w}_0|\nonumber\\
	&\ge&(2\lambda+2)|\bm{w}_0|. \nonumber
\end{eqnarray}
Similarly, for $|\xi_0|<|\eta_0|$, 
\begin{eqnarray}
|\bm{w}_1|&\ge& |\eta_1|\nonumber\\
	&=&|-2Y\xi_0-2(X-c)\eta_0|\nonumber\\
	&\ge&|-2(X-c)||\eta_0|-|-2Y||\xi_0|\nonumber\\
	&=&2|(X-c)||\eta_0|-2|Y||\xi_0|\nonumber\\
	&\ge&2(|X|-c)|\eta_0|-2|Y||\xi_0|\nonumber\\
	&>&2(|X|-|Y|-c)|\eta_0|\nonumber\\
	&=&(|X|-|Y|-c)(|\eta_0|+|\eta_0|)\nonumber\\
	&>&(|X|-|Y|-c)(|\xi_0|+|\eta_0|)\nonumber\\
	&\ge&(|X|-|Y|-c)|\bm{w}_0|\nonumber\\
	&\ge&(2\lambda+2)|\bm{w}_0|. \nonumber
\end{eqnarray}
\end{proof}
Combining Theorem \ref{the:CHHyper} with lemma \ref{lem:CHHyper2}, 
we find the following:
\begin{cor}
\label{cor:uniform_2}
If all points in the non-wandering set $\Omega(f)$, if not empty, 
satisfy the condition 
\begin{eqnarray}
	4+c\le |X|-|Y|,|Z|-|W|\label{eq:CHHyper2}, 
\end{eqnarray}
then $\Omega(F)$ and so $\Omega(f)$ is uniformly hyperbolic. 
\end{cor}
\color{black}


\setcounter{equation}{0}
\section{Proof of Main theorems}
\label{sec:proof_Main_theorem}

\subsection{The case with four symbols in the anti-integrable limit}

\n
{\it Topological horseshoe:} \\
In this section, we provide a sufficient condition for
topological horseshoe and uniform hyperbolicity for the case (A), {\it i.e.}, the case 
around the anti-integrable limit with four symbols. 
First, we consider the situation in the original coordinate $(x,y,z,w)$. 
Using the relation $f^{-1}(f(V_f))=V_f$, we find that 
the region $f(V_f)$ is expressed as 
  \begin{eqnarray}
  \left\{
\begin{array}{l}
|z|\le r, \\
|w|\le r, \\
|a_0-z^2-x+c(z-w)|\le r, \\
|a_1-w^2-y-c(z-w)|\le r. \\
\end{array}
\right.
\label{eq:F0V0}
  \end{eqnarray}
  \color{black}
We can re-express $f(V_f)$ as
\begin{eqnarray}
\nonumber
\fl
f(V_f) = \{ (x,y,z,w) \, | \, |z| \le r, |w|\le r, x =-z^2+cz + a_0 + \alpha ~ \\
\hspace{-10mm}
{\rm where}~
|\alpha| \le (c+1) r, y=-w^2 + cw +a_1+\beta  ~ {\rm where}~|\beta| \le (c+1) r
\}. 
\end{eqnarray}
In this new expression, we have got rid of $w$-dependence of 
$x$ or $z$, 
as well as the 
$z$-dependence of 
$y$ or $w$. 
Therefore the $(x,z)$-plane is now decoupled from the $(y,w)$-plane. 
It is therefore valid to consider parabolas in the $(x,z)$-plane and the $(y,w)$-plane, 
separately.

Let $\Gamma_{x}^{\rm{max}}$ be the parabola with the largest $x$ 
(rightmost
in Fig.~\ref{fig:Hyper1}), 
$\Gamma_{x}^{\rm{min}}$ be the one with the smallest $x$ (leftmost
in 
Fig.~\ref{fig:Hyper1}):
     \begin{eqnarray}
  \Gamma_{x}^{\rm max}: x&=&-z^2+cz+a_0+(c+1)r, \\
  \Gamma_{x}^{\rm min}: x&=&-z^2+cz+a_0-(c+1)r.
  \end{eqnarray}
Furthermore, let 
$S^+_x = \{ (x,y,z,w) \in V \, | \, x = r \}$
and 
$S^-_x = \{ (x,y,z,w) \in V \, | \, x = -r \}$, 
respectively 
(see Fig.~\ref{fig:Hyper1}). 
   \begin{figure}
    \centering
		\includegraphics[height=4cm,bb = 0 0 832 322]{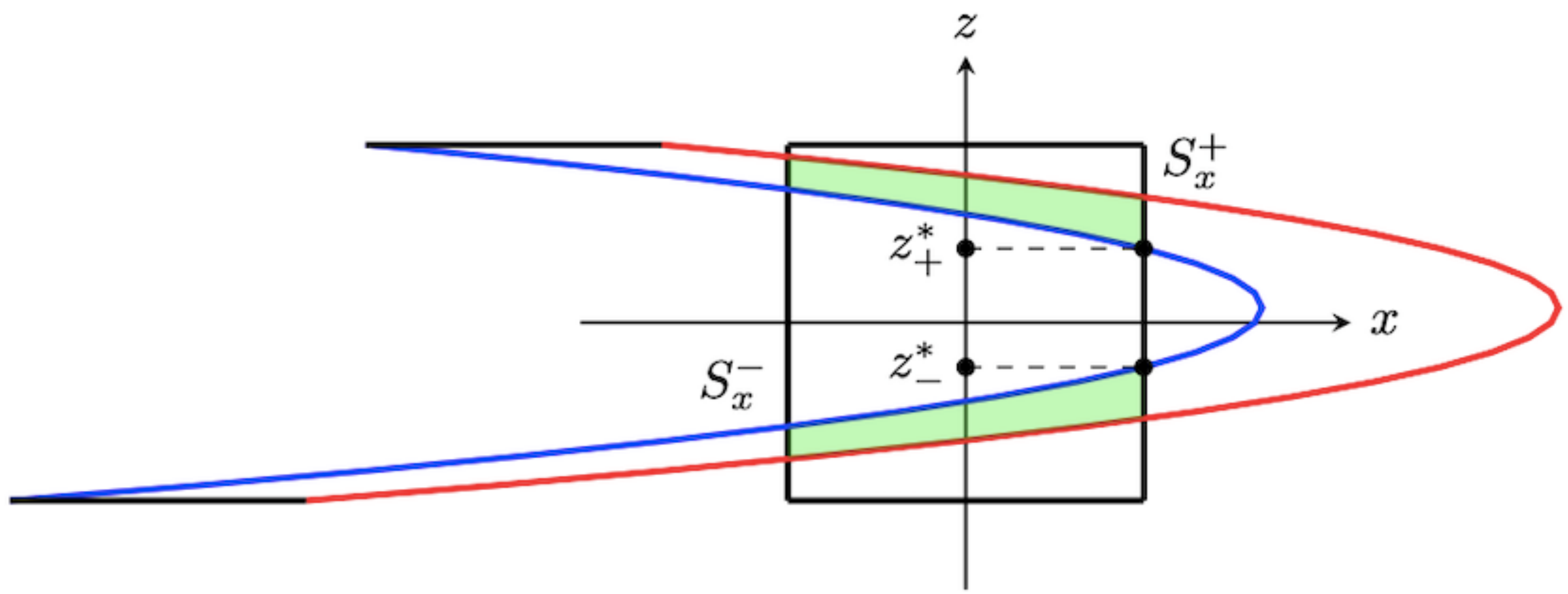}
    \caption{$\Gamma_{x}$ in the $(x,z)$-plane. The red and blue curves 
    represent  $\Gamma_{x}^{\rm{max}}$ and $\Gamma_{x}^{\rm{min}}$, respectively. 
    The green regions show $f(V_f)\cap V_f$. 
    }
    \label{fig:Hyper1}
  \end{figure}
For the 2-dimensional H\'enon map $f$, 
the horseshoe condition is given
by the requirement that $f\cap f(V)$ is decomposed into two disjoint regions.
Here the region $V$ is a region that contains the non-wandering set $\Omega(f)$. 
Here we apply the same condition for the $(x,z)$- and $(y,w)$-planes, respectively. 
First we consider the condition for the $(x,z)$-plane. 
In order for the horseshoe condition to be satisfied in the $(x,z)$-plane, 
as shown in Fig.~\ref{fig:Hyper1}, 
the following should hold: \\
1) 
$\Gamma_{x}^{\rm{min}}$ intersects $S^+_x$ at two points.
\color{black}
\\
2) $\Gamma_{x}^{\rm{max}}$ intersects $S^-_x$ at two points.

The first condition holds if 
\begin{eqnarray}\label{eq:Suff_Cond1}
\frac{1}{4}c^2+a_0-(c+1)r>r 
\end{eqnarray}
is satisfied. Since it is assumed that $c >0$, the second condition is equivalent to 
the condition requiring that 
$x(z=r) \le -r$ and $x(z=-r) \le -r$. 
The former condition is written as
\begin{eqnarray}\label{eq:Suff_Cond2}
-r^2+cr+a_0+(c+1)r\le -r. 
\end{eqnarray}
The latter condition automatically holds if the former one is fulfilled. 
\color{black}

The argument for the $(y,w)$-plane is developed in the same way, again based on (\ref{eq:F0V0}), 
which leads to the conditions 
\begin{eqnarray}\label{eq:Suff_Cond3}
\frac{1}{4}c^2+a_1-(c+1)r>r,
\end{eqnarray}
\begin{eqnarray}\label{eq:Suff_Cond4}
-r^2+cr+a_1+(c+1)r\le -r.
\end{eqnarray}
Due to the symmetry, the inverse map $f^{-1}$ is obtained 
by swapping $(x,y)\leftrightarrow(z,w)$ in the map $f$, 
thus the same conditions follow for $f^{-1}$. 
Thus, the conditions 
(\ref{eq:Suff_Cond1}), (\ref{eq:Suff_Cond2}), (\ref{eq:Suff_Cond3}), and (\ref{eq:Suff_Cond4}) 
lead to a topological horseshoe. 
The proof of Theorem \ref{thm:MainA} A-1) is done. 
\color{black}

\bn
{\it Uniform hyperbolicity:} \\
Next, we consider a sufficient condition for uniform hyperbolicity. 
From section \ref{sec:CHHyper4}, to obtain uniform hyperbolicity it is sufficient to show
 that any point $(x,y,z,w) \in f(V_f)\cap V_f$ satisfies 
the condition (\ref{eq:CHHyper4}) since 
$\Omega(f)\subset f(V_f)\cap V_f$ holds. 

Suppose that $(x,y,z,w) \in f(V_f)\cap V_f$, 
and the conditions (\ref{eq:Suff_Cond1}) and (\ref{eq:Suff_Cond2}) 
are satisfied. 
Let $z^*_-$ and $z^*_+$ be the $z$ coordinates of the intersection points between 
$\Gamma_{\rm{min}}$ and $S_x^+$ where 
$z^*_-\le z^*_+$ is assumed (see Fig.~\ref{fig:Hyper1}). 
We can explicitly obtain as
\begin{eqnarray}
z^*_{\pm} = \frac{c\pm \sqrt{c^2+4(a_0-(c+2)r)}}{2}. 
\end{eqnarray}
If $z^*_- < 0$, the following holds: 
\begin{eqnarray}
	|z|&\ge& \min(|z^*_-|,|z^*_+|)\nonumber = |z^*_-|=-z^*_-. 
\end{eqnarray}
Here $c >0$ is used to show the first 
inequality. 
Hence, if the condition 
\begin{eqnarray}\label{eq:Cond_for_Uniform1}
	- z^*_- >4+c
\end{eqnarray}
\color{black}
is satisfied, then $|z| > 4+c$ holds 
for all the points in $f(V) \cap V$. 
Note that the condition (\ref{eq:Cond_for_Uniform1}) automatically 
ensures the condition $z^*_- < 0$ for $c > 0$. 

We can develop the same argument for the $(y,w)$-plane, and 
find that the following is sufficient to ensure that the condition 
$|w|>4+c$ holds for all the 
points within $f(V) \cap V$:
\begin{eqnarray}
	\frac{-c+\sqrt{c^2+4(a_1-(c+2)r)}}{2}>4+c.
	\label{eq:Cond_for_Uniform2}
\end{eqnarray}
In a similar manner, the argument for the inverse map $f^{-1}$ provides 
a sufficient condition to satisfy $|x|,|y|>4+c$. 
Since the inverse map $f^{-1}$ is given by swapping 
the variables as $(x,y)\leftrightarrow(z,w)$, 
the resulting conditions are the same as (\ref{eq:Cond_for_Uniform1}) 
and (\ref{eq:Cond_for_Uniform2}). 
Thus, in addition to the conditions 
(\ref{eq:Suff_Cond1}), (\ref{eq:Suff_Cond2}), (\ref{eq:Suff_Cond3}), and (\ref{eq:Suff_Cond4}) 
the conditions (\ref{eq:Cond_for_Uniform1}) and (\ref{eq:Cond_for_Uniform2}) 
lead to a sufficient 
condition for the non-wandering set $\Omega(f)$ to be uniformly hyperbolic. 
The proof of Theorem \ref{thm:MainA} A-2) is complete. 
\color{black}

\subsection{The case with two symbols in the anti-integrable limit}
\begin{figure}
        \centering
	\includegraphics[height=10cm,bb = 0 0 1006 843]{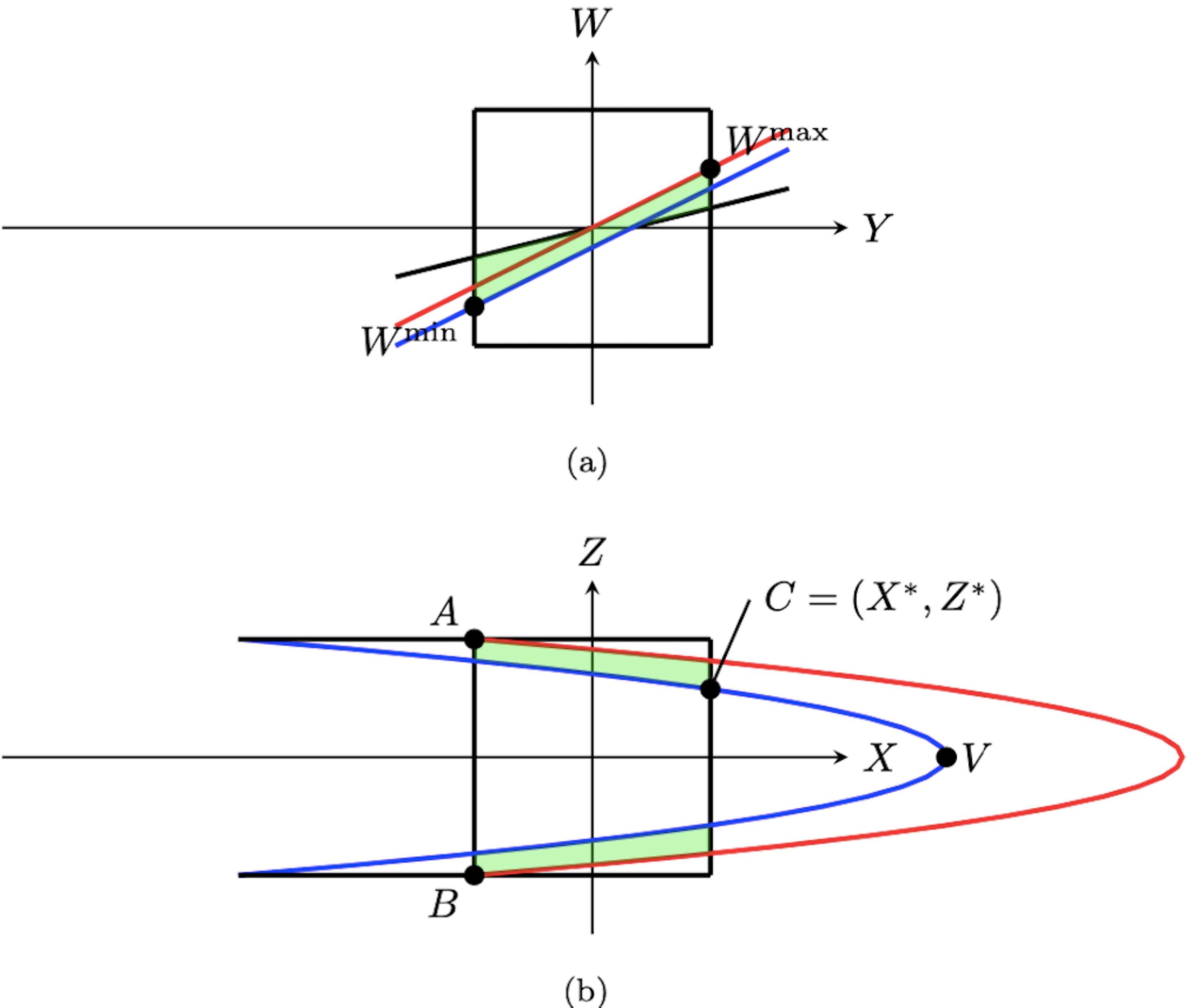}
    \caption{
       (a) Projection of $F(V_F)\cap V_F$ onto the $(Y,W)$-plane (green). 
        The red and blue curves illustrate $\Gamma_Y^{\rm{max}}$ 
    and $\Gamma_Y^{\rm{min}}$, respectively. 
    (b) Projection of $F(V_F)\cap V_F$ onto the $(X,Z)$-plane. 
     The red and blue curves illustrate $\Gamma_X^{\rm{max}}$ 
    and $\Gamma_X^{\rm{min}}$, respectively. 
    }
    \label{fig:CHHyper2}
\end{figure}
\n
{\it Topological horseshoe:} \\
We first examine the existence of topological horseshoe. 
From the definition (\ref{eq:def_V_F}) of $V_F$ and the mapping rule (\ref{eq:map_tranformed}), 
we have 
\begin{eqnarray}
\fl
\nonumber
~~~~~~~F(V_F) \cap V_F = \{(X,Y,Z,W) \, |&& \, |X|, |Y|, |Z|, |W|  \le R, \\
\nonumber
\, &&~X = -Z^2 -W^2 + A_0 + s' ,  |s'| \le R, \\
\, &&~Y = A_1 + 2(c-Z) W + s,  |s| \le R\}. 
\end{eqnarray}
\color{black}
First, consider the projection of $F(V_F) \cap V_F$ onto the $(Y,W)$-plane. 
Let 
\begin{eqnarray}
\Gamma_Y: Y = A_1 -2(c -Z) W +s, 
\end{eqnarray}
be a set of straight lines in the $(Y,W)$-plane parametrized by $Z$ and $s$, 
where $|s| \le R$, 
and let 
\begin{eqnarray}
&&\Gamma_Y^{\rm max}:  Y = A_1 +2(c -R) W - R, \\
&&\Gamma_Y^{\rm min}: Y = A_1 +2(c -R) W + R, 
\end{eqnarray}
be the upper and lower straight members of $\Gamma_Y$. 
$\Gamma_Y^{\rm max}$ is attained at $Z=R$ and 
$s = -R$, and 
$\Gamma_Y^{\rm min}$ is attained at $Z=R$ and $s = R$ (see Fig~\ref{fig:CHHyper2}(a)). 
\color{black}
Since $c > R$ and $|Z| \le R$, we know that the slope of $\Gamma_Y$ 
is always postive. 
Solving for $W$, we get 
\begin{eqnarray}
W = \frac{Y- A_1-s}{2(c-Z)}. 
\end{eqnarray}
The maximum and minimum values of $W$, 
denoted by $W_{\rm max}$ and $W_{\rm min}$ respectively, 
are given as
\begin{eqnarray}
&&W_{\rm max} = \frac{2R-A_1}{2(c-R)}, ~~~{\rm attained~at}~ Y=R, Z=R~ {\rm and}~ s = -R,\\
&&W_{\rm min} = \frac{-2R-A_1}{2(c-R)}, ~~~{\rm attained~at}~ Y=-R, Z=R~ {\rm and}~  s = R.
\end{eqnarray}
Since we have imposed the condition (\ref{eq:typeB_sufficient3}), 
we see that 
\color{black}
the projection of $F(V_F)$ intersects 
$V_F$ completely in the $Y$-direction, and the width of $F(V_F) \cap V_F$, 
as measured in the $W$-direction, is strictly less than $2R$ (see Fig~\ref{fig:CHHyper2}(a)).

Next, consider the projection of $F(V_F) \cap V_F$ onto the $(X,Z)$-plane. 
Let
\begin{eqnarray}
\Gamma_X: X = -Z^2 -W^2 + A_0 + s'
\end{eqnarray}
be a family of parabolas in the $(X,Z)$-plane parametrized by $W$ and $s'$, where 
$|W| \le W^*$ and $|s'| \le R$. 
Let 
\begin{eqnarray}
&&\Gamma_X^{\rm max}: X = -Z^2 + A_0 + R, \\
&&\Gamma_X^{\rm min}: X = -Z^2 + A_0 - (W^*)^2 -R, 
\end{eqnarray}
be the rightmost and leftmost members of $\Gamma_X$. 
Note that 
$\Gamma_X^{\rm max}$ is attained at $W=0$ and $s' = R$, and 
$\Gamma_X^{\rm min}$ at $W=W^*$ and $s' = -R$ (see Fig.~\ref{fig:CHHyper2}(b)). 
For $\Gamma_X^{\rm max}$, notice that when $Z = \pm R$, we have 
\begin{eqnarray}
X = -R^2 + A_0 + R = -R. 
\end{eqnarray}
Therefore, $\Gamma_X^{\rm max}$ intersects with boundary of $V_F$ 
at its two corner points, namely, 
$A = (-R, R)$ and $B =(-R, -R)$ in Fig.~\ref{fig:CHHyper2}(b). 

In the meantime, for $\Gamma_X^{\rm min}$, we examine the location of 
its vertex, denoted by V in Fig.~\ref{fig:CHHyper2}(b). 
The vertex is attained by setting $Z=0$, which leads to 
$$
X_V = A_0 - (W^*)^2 -R. 
$$
Since it is imposed in (\ref{eq:typeB_sufficient2}) that 
$$
A_0 - (W^*)^2 -R > R, 
$$
we obtain $X_V > R$, i.e., the vertex of $\Gamma_X^{\rm min}$ is located on the 
right side of $(R,0)$, as illustrated in Fig.~\ref{fig:CHHyper2}(b). 

As a result, the region in between $\Gamma_X^{\rm max}$ and $\Gamma_X^{\rm min}$ 
gives rise to a topological binary horseshoe in the $(X,Z)$-plane.
Thus, we know that the non-wandering set $\Omega(F)$ is non-empty and 
is at least semi-conjugate to a full shift with two symbols.

\bn
{\it Uniform hyperbolicity:} \\
Next, we will show uniform hyperbolicity on $\Omega(F)$. 
From section \ref{sec:uniform_hyp_2}, we already know a sufficient condition 
for uniform hyperbolicity in Corollary \ref{cor:uniform_2}. 
\color{black}
Here we show that this is indeed the case for points in $F(V_F) \cap V_F$. 

Notice that for any point in $F(V_F) \cap V_F$, we have 
\begin{eqnarray}
|Z| \ge Z^*, 
\end{eqnarray}
where $Z^*$ is the $Z$-coordinate of the point $C$ in Fig~\ref{fig:CHHyper2}(b). 
\color{black}
Thus, 
\begin{eqnarray}
|Z| - |W| \ge Z^* - |W| \ge Z^* - W^* 
\end{eqnarray}
holds. Since it is imposed 
in (\ref{the:CHBe}) 
\color{black}
that $Z^* - W^* \ge 4 + c$, we immediately obtain
\begin{eqnarray}
|Z| - |W|  \ge 4 + c.
\end{eqnarray}

Due to the symmetry of the mapping equations, 
$F^{-1}$ can be obtained from $F$ by swapping $(X,Z)$ with $(Y,W)$, 
thus we obtain, 
\begin{eqnarray}
|X|- |Y| \ge  4 + c 
\end{eqnarray}
as well.  The uniform hyperbolicity on $\Omega(F)$ thus follows. 

\begin{figure}
 \centering
	\includegraphics[width=9cm,bb = 0 0 1200 900]{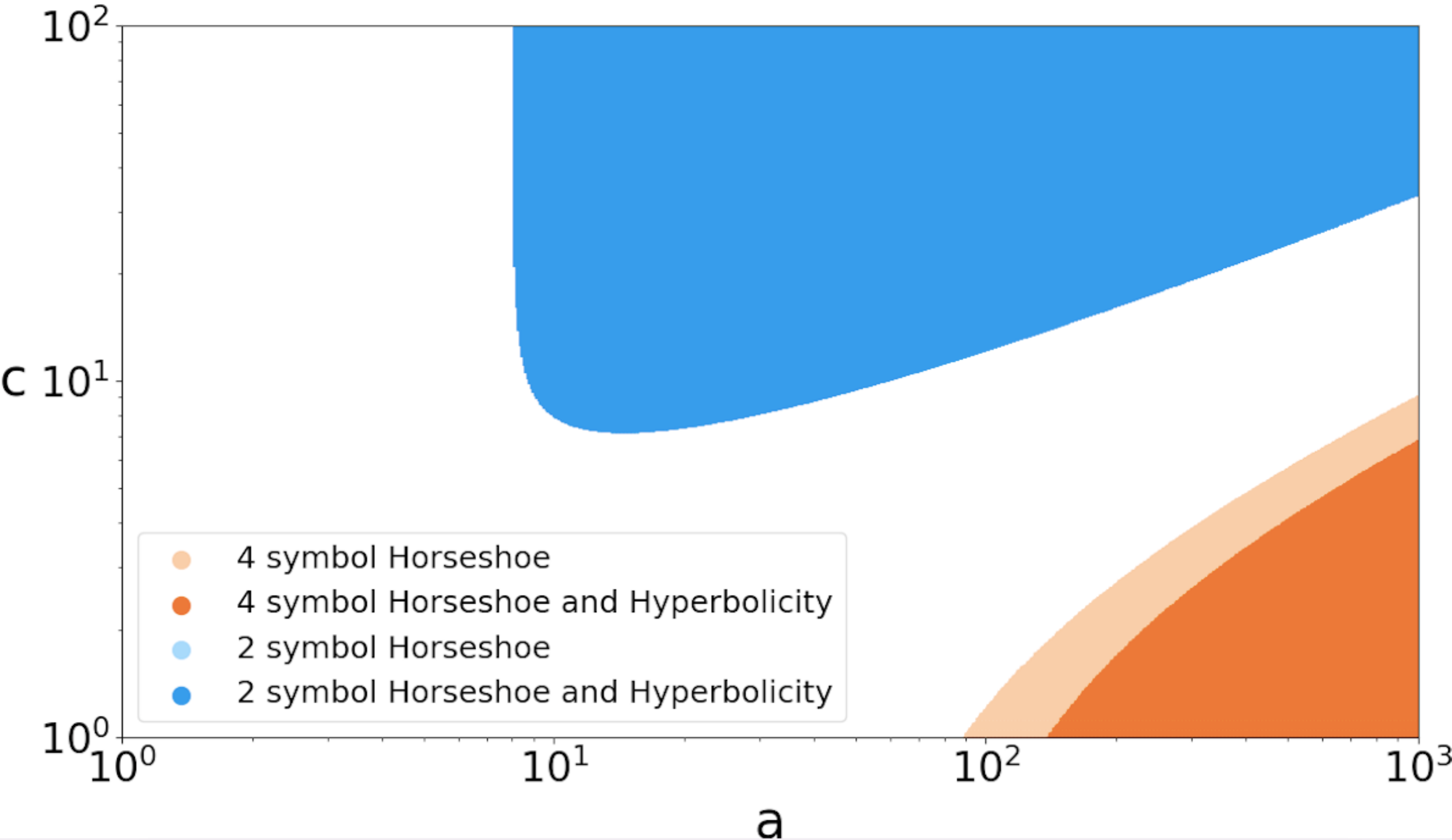}
 \caption{For the anti-itegrable limit with four symbols, 
the region satisfying the topological horseshoe is shown in light orange, 
and the region satisfying both topological horseshoe and uniform hyperbolicity is 
shown in orange.  
For the the anti-itegrable limit with two symbols, 
the region satisfying topological horseshoe is shown in light blue, 
and the region satisfying both topological horseshoe and uniform hyperbolicity is 
shown in blue.  $a=a_0=a_1$ are taken. 
}
 \label{fig:CHPara}
\end{figure}

Finally, we check that the parameters leading to the anti-integrable limit satisfy 
the sufficient condition obtained above for topological horseshoe and uniform 
hyperbolicity. 
The case (A) is given by taking the limit of $a=a_0=a_1 \to \infty$. 
This limit implies that $r\rightarrow 2\sqrt{2a}$, so 
it turns out that the conditions in A-2) and A-3) in Theorem \ref{the:CHA} hold. 
For the case (B), the anti-integrable limit is obtained 
by taking the limit of $a=a_0=a_1\rightarrow \infty$ and $\gamma\rightarrow \infty$ 
with $c=\gamma\sqrt{a}$ being fixed. 
In this case, $R\rightarrow\sqrt{a}$, $W^*\rightarrow 0$ and $Z^*=\sqrt{a}$ follow, 
and the conditions in B-2) and B-3) in Theorem \ref{the:CHB} are satisfied. 
Figure~\ref{fig:CHPara} illustrates the parameter regions 
in which topological horseshoe and uniform 
hyperbolicity hold.

\setcounter{equation}{0}
\section{Summary}
\label{sec:summary}

We have derived a sufficient condition for topological horseshoe and 
uniform hyperbolicity of  the coupled \He map around the anti-integrable limits. 
The coupled \He map introduced here has at least two types of anti-integrable limits, 
which were obtained by taking appropriate limits on the nonlinear parameters $a_0$, $a_1$ and 
a coupling strength $c$. 
The strategy of specifying the existence domain of the non-wandering set, and showing 
topological horseshoe and uniform hyperbolicity is a straightforward generalization 
of the approach taken in Ref.~\cite{DN79}. 
It is specific to higher dimensional maps to have different types of horseshoe, and
it does not happen in 2-dimensional maps. In a subsequent paper \cite{Li23},  we will further introduce topologically different types of 
horseshoe that are impossible in two dimensions by studying a family of H\'enon-type mappings. 

Since the conditions obtained are sufficient ones, as in the case of 
the 2-dimensional \He map~\cite{DN79}, 
one can expect that the parameter domain with topological horseshoe and 
uniform hyperbolicity must be further extended, possibly to 
the situation where an analog of the first tangency happens~\cite{Sterling99,Bedford04}. 
A plausible approach to this problem would be to use 
a computer-assisted proof developed in Refs.~\cite{Arai07,Arai07_2}.
Furthermore, it is interesting to investigate the transition between the two types of 
horseshoes found in the present work. 
Such a transition, if it exists, will induce a kind of bifurcation in higher dimensions. 

Another question to be addressed in the future is 
whether other types of horseshoes exist in the parameter space. 
We have studied here only in the symmetric situation $a_0 = a_1$, but it is by no means obvious 
whether the situation associated with three symbols appears or not.  
If this is the case, this also provides a new type of horseshoe, which appears only 
in higher dimensional maps.

\color{black}

\section*{Acknowledgement}
J.L. and A.S. acknowledge financial support from Japan Society for the Promotion of Science (JSPS) through JSPS Postdoctoral Fellowship for Research in Japan (Standard). 
This work has been supported by JSPS KAKENHI Grant No. 17K05583, and also by JST, the establishment of university fellowships towards the creation of science technology innovation, Grant Number JPMJFS2139.

\section*{References}



\begin{thebibliography}{99}




\bibitem{Henon69}
H{\'e}non M 1969
Numerical study of quadratic area-preserving mappings
{\it Quarterly of applied mathematics}
291--312

\bibitem{Henon76}
H{\'e}non M 1976
A two-dimensional mapping with a strange attractor
{\it Communications in Mathematical Physics}
{\bf 50}
69-77

\bibitem{Katok95}
Katok A and Hasselblatt B 1995
{\it Introduction to the modern theory of dynamical systems} 
(Cambridge: Cambridge University Press)

\bibitem{Wiggins88}
Wiggins S 1988
{\it Global bifurcations and chaos: analytical methods}
(Springer-Verlag)


\bibitem{Wiggins03}
Wiggins S 2003
{\it Introduction to applied nonlinear dynamical systems and chaos}
(Springer)


\bibitem{DN79}
Devaney R and Nitecki Z 1979
Shift automorphisms in the H{\'e}non mapping
 {\it Communications in Mathematical Physics}
 {\bf 67}
137--146


\bibitem{Bedford04}
Bedford E and Smillie J 2004
Real polynomial diffeomorphisms with maximal entropy: Tangencies
 {\it Annals of mathematics}
 {\bf 160}
1--26

\bibitem{Arai07}
Arai Z 2007
On hyperbolic plateaus of the H{\'e}non map
{\it Experimental Mathematics}
{\bf 16}
181--188

\bibitem{Arai07_2}
Arai Z 2007
On Loops in the Hyperbolic Locus of the Complex H\'enon Map and Their Monodromies
{\it  arXiv preprint arXiv:0704.2978} 


\bibitem{Aubry90}
Aubry S and Abramovici G 1990
Chaotic trajectories in the standard map. The concept of anti-integrability
{\it Physica D: Nonlinear Phenomena}
{\bf 43}
199--219

\bibitem{Aubry94}
Aubry S 1994
The concept of anti-integrability applied to dynamical systems and to structural and electronic models in condensed matter physics
{\it Physica D: Nonlinear Phenomena}
{\bf 71 }
196--221

\bibitem{Bolotin15}
Bolotin S V and Treschev D V 2015
The anti-integrable limit
{\it Russian Mathematical Surveys}
{\bf 70 }
975--1030


\bibitem{Aubry92}
Aubry S, MacKay R S and Baesens C
quivalence of uniform hyperbolicity for symplectic twist maps and phonon gap for Frenkel-Kontorova models
{\it Physica D: Nonlinear Phenomena}
{\bf 56}
123--134





30



\bibitem{Mao88}
Mao J 1988
Standard form of four-dimensional symplectic quadratic maps
{\it Physical Review A}
{\bf 38}
525--526

\bibitem{Howard87}
Howard J E and MacKay R S 1987
Linear stability of symplectic maps
{\it Journal of mathematical physics}
{\bf 5}
1036--1051

\bibitem{Ding90}
Ding M, Bountis T and Ott E 1990
Algebraic escape in higher dimensional Hamiltonian system
{\it Physics Letters A}
{\bf 151}
395--400

\bibitem{Bountis94}
Bountis T and Kollmann M 1994
Diffusion rates in a 4-dimensional mapping model of accelerator dynamics
{\it Physica D: Nonlinear Phenomena}
{\bf 71}
122--131

\bibitem{Todesco94}
Todesco E 1994
Analysis of resonant structures of four-dimensional symplectic mappings, using normal forms
{\it Physical Review E}
{\bf 50}
R4298--R4302

\bibitem{Vrahatis96}
Vrahatis M N,  Bountis T C and Kollmann M 1996
Periodic orbits and invariant surfaces of 4D nonlinear mappings
{\it International Journal of Bifurcation and Chaos}
{\bf 6}
1425--1437

\bibitem{Todesco96}
Todesco E 1996
Local analysis of formal stability and existence of fixed points in 4d symplectic mappings
{\it Physica D: Nonlinear Phenomena}
{\bf 95}
1--12

\bibitem{Gemmi97}
Gemmi M and Todesco E 1997
Stability and geometry of third-order resonances in four-dimensional symplectic mappings
{\it Celestial Mechanics and Dynamical Astronomy}
{\bf 67}
181--204

\bibitem{Vrahatis97}
Vrahatis M N, Isliker H and Bountis, T C 1997
Structure and breakdown of invariant tori in a 4-D mapping model of accelerator dynamics
{\it International Journal of Bifurcation and Chaos}
{\bf 7}
2707--2722

\bibitem{giovannozzi98}
Giovannozzi M, Scandale W and Todesco E 1998
Dynamic aperture extrapolation in the presence of tune modulation
{\it Physical Review E}
{\bf 57}
3432--3443

\bibitem{Richter14}
Richter M, Lange S, B{\"a}cker A and Ketzmerick R 2014
Visualization and comparison of classical structures and quantum states of four-dimensional maps
{\it Physical Review E}
{\bf 89}
022902 

\bibitem{Lange14}
Lange S, Richter M, Onken F, B{\"a}cker A and Ketzmerick R 2014
Global structure of regular tori in a generic 4D symplectic map
{\it Chaos: An Interdisciplinary Journal of Nonlinear Science}
{\bf 24}
024409 

\bibitem{Onken16}
Onken F, Lange S, Ketzmerick R and B{\"a}cker, A  2016
Bifurcations of families of 1D-tori in 4D symplectic maps
{\it Chaos: An Interdisciplinary Journal of Nonlinear Science}
{\bf 26}
063124 

\bibitem{Lange16}
Lange S, B{\"a}cker A and Ketzmerick, R
What is the mechanism of power-law distributed Poincar{\'e} recurrences in higher-dimensional systems?
{\it EPL (Europhysics Letters)}
{\bf 116 }
30002

\bibitem{Anastassiou17}
Anastassiou S, Bountis T and B{\"a}cker, A 2017
Homoclinic points of 2D and 4D maps via the parametrization method
{\it Nonlinearity}
{\bf 30}
3799
%
%
%



\bibitem{Backer18}
B\"acker A and Meiss J E 2018
Moser's quadratic, symplectic map
{\it Regular and Chaotic Dynamics}
{\bf 23}
654--664

\bibitem{Backer20}
B\"acker A and Meiss J E 2020
Elliptic Bubbles in Moser's 4D Quadratic Map: The Quadfurcation 
{\it SIAM Journal on Applied Dynamical Systems}
{\bf 19}
442--479



\bibitem{Friedland89}
Friedland S and Milnor J 1989
Dynamical properties of plane polynomial automorphisms
{\it Ergodic Theory and Dynamical Systems}
{\bf 8}
67--99

\bibitem{Moser94}
Moser J 1994
On quadratic symplectic mappings
{\it Mathematische Zeitschrift}
{\bf 216}
417--4

\bibitem{Li23}
Li J, Fujioka K, and Shudo A 2023
{Coupled H\'{e}non Map, Part II: Doubly and Singly Folded Horseshoes in Four Dimensions,}
{\it to be submitted}


\bibitem{Qin01}
Qin W X 2001
Chaotic invariant sets of high-dimensional H{\'e}non-like maps
{\it Journal of mathematical analysis and applications}
{\bf 264}
76--84

\bibitem{Aubry05}
Aubry S 1995
Anti-integrability in the dynamical and variational problems
{\it Physica D} {\bf 86} 284-296 

\bibitem{Du06}
Du B, Li, M C and Malkin M I 2006
Topological horseshoes for Arneodo Coullet Tresser maps
{\it Regular and Chaotic Dynamics}
{\bf 11}
181--190
%

\bibitem{Li06}
Li M C and Malkin M 2006
Topological horseshoes for perturbations of singular difference equations
{\it Nonlinearity}
{\bf 19}
795--811

\bibitem{Juang08}
Juang J, Li, M C and Malkin M 2008
Chaotic difference equations in two variables and their multidimensional perturbations
{\it Nonlinearity}
{\bf 21}
1019-1040

\bibitem{Chen15}
Chen H J and Li M C 2015
Stability of symbolic embeddings for difference equations and their multidimensional perturbations
{\it Journal of Differential Equations}
{\bf 258}
906--918

\bibitem{Chen16}
Chen T H, Lin, W W and Peng C C 2016
Chaotic orbits for differentiable maps near anti-integrable limits
{\it Journal of Mathematical Analysis and Applications}
{\bf 535}
889--916

\bibitem{Hampton22}
Hampton A E and Meiss J D 2022
Anti-integrability for Three-Dimensional Quadratic Maps
{\it SIAM Journal on Applied Dynamical Systems}
{\bf 21}
650--675




%
%
%






\bibitem{Newhouse04}
Newhouse R  2004
Cone-fields, domination and hyperbolicity
{\it Modern Dynamical Systems and Applications}
(Cambridge University Press) 
419-432



%



\bibitem{Sterling99}
Sterling D, Dullin H R and Meiss J D 1999
Homoclinic bifurcations for the H{\'e}non map
{\it Physica D: Nonlinear Phenomena}
{\bf 134}
153--184

\color{black}





\end{thebibliography}
\end{document}